\documentclass[12pt]{amsart}
\pdfoutput=1
\usepackage{amsmath,amsthm,amssymb,latexsym} % AMS-LaTeX + Theorems + Symbols
\usepackage{fullpage} % decreases margins
\usepackage{enumerate} % easy customization of enumeration
\usepackage{graphicx} % to include graphics files
\graphicspath{{./}{./pictures/}} % temporary path to images
\usepackage{color} % adds color
\usepackage{accents} % for \undertilde
\usepackage[nodate,hhmmss]{datetime}
\usepackage[pdftex,bookmarks,colorlinks,breaklinks]{hyperref}  % PDF hyperlinks, with coloured links
\definecolor{dullmagenta}{rgb}{0.4,0,0.4}   % #660066
\definecolor{darkblue}{rgb}{0,0,0.4}
\hypersetup{linkcolor=red,citecolor=blue,filecolor=dullmagenta,urlcolor=darkblue} % coloured links
\newtheorem{theorem}{Theorem}[section]
\newtheorem{lemma}[theorem]{Lemma}

\theoremstyle{definition}
\newtheorem{definition}[theorem]{Definition}

\theoremstyle{remark}

\numberwithin{equation}{section}

\begin{document}

\title[Combinatorics of Factorizations]{Combinatorics of Matrix Factorizations and Integrable Systems}

% Author Informations	
\author{Anton Dzhamay}
\address{School of Mathematical Sciences\\ 
The University of Northern Colorado\\ 
Campus Box 122\\ 
501 20th Street\\ 
Greeley, CO 80639, USA}
\email{\href{mailto:adzham@unco.edu}{\texttt{adzham@unco.edu}}}

\begin{abstract} We study relations between the eigenvectors of rational
	matrix functions on the Riemann sphere. Our main result is that for a 
	subclass of functions that are products of two elementary 
	blocks it is possible to represent these  
	relations in a combinatorial--geometric way using a
	diagram of a cube. In this representation, vertices of the cube represent 
	eigenvectors, edges are labeled by differences of locations of zeroes and poles of 
	the determinant of our matrix function, and each face corresponds to a 
	particular choice of a coordinate system on the space of such functions.
	Moreover, for each face this labeling encodes, in a neat and efficient way,
	a generating function for the expressions of the remaining four eigenvectors that 
	label the opposing face of the cube in terms of the coordinates
	represented by the chosen face.	The main motivation behind this work is that when our matrix 
	is a Lax matrix of a discrete integrable system, such generating functions can be interpreted as
	Lagrangians of the system, and a choice of a particular face corresponds to a choice of the 
	direction of the motion.
	% 
	% 
	%  A particular example of a Lagrangian of this type is the generating 
	% function describing isomonodromic transformations of linear \emph{difference} equations.
	% Similarly to the continuous case, these isomonodromic transformations admit reductions 
	% described by \emph{difference} 
	% Such tranThese transformations re
	% 
	% In particular, the Lagrangian describing isomonodromic 
	% transformations of , that we have obtained recently, is exactly of 
	% this type. 
	% A particularly interesting example 
	% In particular, the Lagrangian of isomonodromic transformations that are described by 
	% difference Painlev\'e V equations 
	%  that we have found previously we have previously shown that  that reduce to 
	% d correspond to exactly this type of re-factorization transformations.
\end{abstract}

\maketitle

\section{Introductions} % (fold)
\label{sec:introductions}
Over the last 25 years a lot of research efforts have been directed towards the study of discrete 
analogues of integrable systems and, in particular, on how to adapt the existing methods and
techniques from the classical theory of differential completely integrable systems to the 
discrete case. The present paper is a small part of a larger project that aims to understand the Lagrangian 
structure of discrete integrable systems directly in terms of their Lax matrices and is 
motivated by work of Krichever and Phong \cite{KriPho:1998:SFITTOS}, 
who obtained expressions for a universal symplectic form and elementary generating
Hamiltonians on the space of Lax matrices for continuous completely integrable systems.
Note that a universal formula for a Lagrangian description of integrable systems
in the continuous case was obtained by Zakharov and Mikhailov 
\cite{ZahMih:1980:AVPFETAIBTIPM}.

For discrete completely integrable systems, the Lagrangian point of view 
has been developed in the classical papers by Veselov 
\cite{Ves:1988:ISWDTDO,Ves:1991:ILRFMP}  and Moser and Veselov \cite{MosVes:1991:DVSCISFMP}.
Recently a very promising approach to the study of Lagrangian structure of 
integrable lattice equations in terms of Lagrangian multiforms
has been proposed by Nijhoff and Lobb \cite{LobNij:2009:LMAMC}. This approach is related
to the notion of multidimensional consistency formulated by Bobenko and Suris \cite{BobSur:2002:ISOQ} and
independently by Nijhoff \cite{Nij:2002:LPFTALKS}. Although our approach is much more elementary and 
is based on the notion of the \emph{refactorization transformations}, as in the original work of Veselov and Moser,
the resulting combinatorial diagrams encoding such transformations 
are reminiscent of the multidimensional consistency approach, and it will be very interesting to see
if this is more than a pure coincidence. 

The appearance of refactorization transformations in the theory of discrete integrable systems is not 
very surprising. Indeed, a discrete analogue of the Lax Pair representation is the isospectral 
transformation $\tilde{\mathbf{L}}(z) = \mathbf{R}(z) \mathbf{L}(z) \mathbf{R}(z)^{-1}$, where 
$\mathbf{L}(z)$ is the Lax matrix of the system and $\mathbf{R}(z)$ is the evolution matrix that has to
be chosen in a special way dependent on $\mathbf{L}(z)$. When $\mathbf{L}(z)$ is a rational matrix
function with the fixed singularity structure, one can specify elementary evolution matrices $\mathbf{R}(z)$
using pairs of points $p_{\pm}$ on the spectral curve $\Gamma$ of $\mathbf{L}(z)$ and the eigenvectors
of $\mathbf{L}(z)$ at those points, and choosing those points to lie above the points of the determinantal
divisor of $\det \mathbf{L}(z)$ corresponds to factoring $\mathbf{L}(z) = \mathbf{L}_{1}(z) \mathbf{B}_{2}(z)$ and mapping
it to $\tilde{\mathbf{L}}(z) = \mathbf{B}_{2}(z) \mathbf{L}_{1}(z)$, i.e., by choosing $\mathbf{R}(z) = \mathbf{B}_{2}(z)$. 

In describing such refactorization transformations it is essential to choose a good coordinate 
system on the space of Lax matrices. The natural candidates for the coordinates are the eigenvectors
of $\mathbf{L}(z)$ and $\mathbf{L}(z)^{-1}$. However, for the refactorization transformations another 
important class of vectors consists of vectors defining elementary building blocks. The relationship
between these two sets of vectors is quite complicated, and in attempt to understand it we noticed a 
very elegant way to represent it using geometric diagrams described in the present paper. Although 
we restrict our attention to the quadratic case when $\mathbf{L}(z)$ is a product of two blocks and when 
there is no essential difference between 
the vectors of the additive and multiplicative representation, it is still quite interesting.  Indeed, such
refactorization transformations describe, for example, the change of polarization in the interaction
of soliton solutions of the matrix KdV equations \cite{GonVes:2004:YMAMS} and integrable 
discrete vector nonlinear Schr\"odinger equations \cite{PriAblBio:2006:ISTFTVNSEWNBC}. Further, in 
\cite{Dzh:2009:FORMF} we showed that for two-dimensional matrix functions the directional dynamics
of the eigenvectors is described by difference Painlev\'e-V equation.

We plan to use this tools developed in this paper to study the higher dimensional case of products of three and more elementary blocks 
and, in particular, the relation to Yang-Baxter maps, in a separate publication.

This paper is organized as follows. In the next section we briefly describe a particular class of Lax matrices that we consider,
their additive and multiplicative representations 
and the notion of an elementary divisor --- a building block for a multiplicative representation. In Section 3 we 
explain a visual representation of some linear equations involving elementary divisors and in Section 4 we use this 
visual representation to describe the relations between eigenvectors of a quadratic Lax matrix, and how to 
use the resulting cube diagram to generate the Lagrangian functions of the refactorization transformations. 

% section introductions (end)

\section{Preliminaries} % (fold)
\label{sec:preliminaries}

The main goal of this section is to describe the space of Lax matrices,
 their additive and multiplicative representations and related coordinate
systems, and briefly review the
relationship between isomonodromic transformations and  discrete Painlev\'e equations,
see \cite{Dzh:2009:FORMF} for details.

\subsection{The space of Lax Matrices} % (fold)
\label{sub:the_space_of_lax_matrices}

We consider the space $\mathcal{L}$ of rational $m\times m$ matrix functions  
$\mathbf{L}(z)$ satisfying the following conditions. 
We assume that $\mathbf{L}(z)$ is regular, diagonalizable (and diagonalized) at some normalization point $z_{0}$
that we take to be $z_{0}=\infty$,
\begin{equation}
	\mathbf{L}_{0} = \lim_{z\to\infty} \mathbf{L}(z) = 
	\operatorname{diag}\{\rho_{1},\dots, \rho_{m}\}.
\end{equation}
We also assume that the singularity structure of $\mathbf{L}(z)$, i.e., points where $\mathbf{L}(z)$ has a pole or
becomes degenerate, is accurately represented by its determinant, and that the determinant is generic, i.e., it has 
only \emph{simple} zeroes and simple \emph{poles}. This amounts to requiring that $\mathbf{L}(z)$ is holomorphic except for  
\emph{simple} poles at the points $z_{1},\dots, z_{k}$, $\mathbf{M}(z) = \mathbf{L}(z)^{-1}$ is holomorphic except for 
 \emph{simple} poles at the points $\zeta_{1},\dots, \zeta_{k}$,  all $z_{i}$ and $\zeta_{j}$ 
are distinct, and 
\begin{equation}
	\det \mathbf{L}(z)=\rho_{1}\cdots \rho_{m}\frac{\prod_{\alpha}(z-\zeta_{\alpha})}{\prod_{k}(z-z_{k})}.
\end{equation}

These conditions mean that the residues $\mathbf{L}_{i}=\operatorname{res}_{z_{i}}\mathbf{L}(z)$
and $\mathbf{M}_{j}=-\operatorname{res}_{\zeta_{j}}\mathbf{M}(z)$, where the negative sign here is just for convenience,
are matrices of rank one. Using the $\dag$ symbol to indicate a row vector, 
we have:
\begin{align}
	\mathbf{L}(z) &= \mathbf{L}_{0} + \sum_{i=1}^{k} \frac{\mathbf{L}_{i}}{z-z_{i}},
	\qquad\text{where } \mathbf{L}_{0}=\operatorname{diag}\{\rho_{1},\dots,\rho_{m}\}\text{ and }
	\mathbf{L}_{i} = \mathbf{a}_{i} \mathbf{b}^{\dag}_{i},\label{L(z)-props-a}\\
	\det \mathbf{L}(z) &= \rho_{1}\cdots \rho_{m} 
	\frac{\prod_{i=1}^{k} (z-\zeta_{i})}{\prod_{j=1}^{k} (z-z_{j})},\label{L(z)-props-b}\\
	\mathbf{L}(z)^{-1}= \mathbf{M}(z) &= \mathbf{M}_{0} - \sum_{i=j}^{k}\frac{\mathbf{M}_{j}}{z - \zeta_{j}},
	\qquad\text{where } \mathbf{M}_{0}=\mathbf{L}_{0}^{-1},\qquad
	\mathbf{M}_{j} = \mathbf{c}_{j} \mathbf{d}^{\dag}_{j}\label{L(z)-props-c},
\end{align}
% and the negative signs in $\mathbf{M}(z)$ are for future convenience. 
The above representations of $\mathbf{L}(z)$ and 
$\mathbf{M}(z)$ are called \emph{additive representations} and the vectors $\mathbf{a}_{i}$, $\mathbf{b}_{i}^{\dag}$ 
(resp. $\mathbf{c}_{i}$, $\mathbf{d}_{i}^{\dag}$) \emph{additive eigenvectors} of 
$\mathbf{L}(z)$ (resp. $\mathbf{M}(z)$). Note that these eigenvectors are also characterized by
$\mathbf{L}(\zeta_{\alpha})\mathbf{c}_{\alpha} = \mathbf{M}(z_{i}) \mathbf{a}_{i} = \mathbf{0}$ and 
$\mathbf{d}_\alpha^{\dag}\mathbf{L}(\zeta_{\alpha}) = \mathbf{b}_{i}^{\dag} \mathbf{M}_{i} = \mathbf{0}$.

Let $\mathcal{D} = \sum_{i} z_{i} - \sum_{i} \zeta_{i}$ be the divisor of the (determinant) of $\mathbf{L}(z)$
and denote the space of matrices $\mathbf{L}(z)$ satisfying conditions (\ref{L(z)-props-a})--(\ref{L(z)-props-c}) by 
$\mathcal{M}_{r}^{\mathcal{D}}$. Then eigenvectors give coordinates on this space, as described by the following 
lemma.

\begin{lemma}\label{lem:eigen-coords}
	Generically, the collection $\{\mathbf{a}_{k}, \mathbf{d}_{k}^{\dag}\}_{k=1}^{n}$ (or the collection
	$\{\mathbf{c}_{k}, \mathbf{b}_{k}^{\dag}\}_{k=1}^{n}$) is a coordinate system on the space $\mathcal{M}_{r}^{\mathcal{D}}$.
\end{lemma}		

\begin{proof}
	Consider the equations $\mathbf{M}(z_{k}) \mathbf{a}_{k} = \mathbf{0}$ and 
	$\mathbf{d}_{i}^{\dag} \mathbf{L}(\zeta_{i}) = \mathbf{0}$:
	\begin{equation}
		\mathbf{L}_{0}^{-1} \mathbf{a}_{k} - \sum_{i=1}^{n} \mathbf{c}_{i} \frac{\mathbf{d}_{i}^{\dag} 
		\mathbf{a}_{k}}{z_{k} - \zeta_{i}} = \mathbf{0}, \qquad
		\mathbf{d}_{i}^{\dag}\mathbf{L}_{0}  + \sum_{k=1}^{n} \frac{\mathbf{d}_{i}^{\dag} 
		\mathbf{a}_{k}}{\zeta_{i} - z_{k}} \mathbf{b}_{k}^{\dag} = \mathbf{0}.			
	\end{equation}
	Then if the matrix $\left[ \frac{\mathbf{d}_{i}^{\dag} \mathbf{a}_{k}}{z_{k} - \zeta_{i}} \right]$ is invertible,

	\begin{equation}
		\mathbf{c}_{i} =  \mathbf{L}_{0}^{-1}\mathbf{a}_{k} \left[ \frac{\mathbf{d}_{i}^{\dag} 
		\mathbf{a}_{k}}{z_{k} - \zeta_{i}} \right]^{-1}, \quad
		\mathbf{b}_{k}^{\dag} = \left[ \frac{\mathbf{d}_{i}^{\dag} \mathbf{a}_{k}}{z_{k} - \zeta_{i}} \right]^{-1} 
		\mathbf{d}_{i}^{\dag} \mathbf{L}_{0}.
	\end{equation}	
\end{proof}

% subsection the_space_of_lax_matrices (end)

\subsection{Elementary Divisors and  Multiplicative Representations} % (fold)
\label{sub:elementary_divisors_and_multiplicative_representations}
To define a multiplicative representation of $\mathbf{L}(z)$ we first define its building blocks. These are 
rational matrix functions of the special form 
\begin{equation}
	\mathbf{B}(z) = \mathbf{I} + \frac{ \mathbf{G} }{ z - z_{0} }, \qquad\text{where $\mathbf{G} = \mathbf{f} \mathbf{g}^{\dag}$ 
	is a matrix of rank 1.}\label{eq:el-div}
\end{equation}
We call such matrices \emph{elementary divisors} \cite{Dzh:2009:FORMF}. If 
$\operatorname{tr}(\mathbf{G}) = \mathbf{g}^{\dag} \mathbf{f}\neq0$, we can 
consider instead of $\mathbf{G}$ a rank-one projector 
$\mathbf{P} = \mathbf{f} (\mathbf{g}^{\dag} \mathbf{f})^{-1} \mathbf{g}^{\dag}$. Then we can write
\begin{equation}
\mathbf{B}(z) = \mathbf{B}(z;z_{0},\zeta_{0}) = 
\mathbf{I} + \frac{ z_{0} - \zeta_{0}}{ z-z_{0}  } \frac{ \mathbf{f} \mathbf{g}^{\dag} }{ \mathbf{g}^{\dag} \mathbf{f} },\quad
\det \mathbf{B}(z) = \frac{ z - \zeta_{0} }{ z-z_{0} },\quad
\mathbf{B}(z)^{-1} = 
	\mathbf{I} + \frac{\zeta_{0} - z_{0}}{z-\zeta_{0}} \frac{ \mathbf{f} \mathbf{g}^{\dag} }{ \mathbf{g}^{\dag} \mathbf{f} }.
\end{equation}

If we now pair the zeroes and poles of $\det \mathbf{L}(z)$ in some way as
$(\zeta_{s},z_{s})$, then, for any such pair, there is a multiplicative component of
$\mathbf{L}(z)$ of the above form. There are two ways to think about such components --- 
we can look at  \emph{factors} or at \emph{divisors}, see also \cite{Bor:2004:ITLSDE}.

\begin{definition}\label{def:divisors} We say that elementary divisors $\mathbf{B}^{r}_{s}(z)$ (resp. 
	$\mathbf{B}^{l}_{s}(z)$) corresponding to pairs $(\zeta_{s},z_{s})$ are \emph{right} (resp. \emph{left})
	\emph{divisors} of $\mathbf{L}(z)$ if $\mathbf{L}(z) = \mathbf{L}^{r}_{s}(z) \mathbf{B}^{r}_{s}(z)$ 
	(resp. $\mathbf{L}(z) = \mathbf{B}^{l}_{s}(z)\mathbf{L}^{l}_{s}(z)$)
	where $\mathbf{L}^{r}_{s}(z)$ (resp. $\mathbf{L}^{l}_{s}(z)$) is regular at $z_{s}$.	
	Further, we say that elementary divisors 
		$\mathbf{B}_{s}(z)$ corresponding to pairs 
		$(\zeta_{s},z_{s})$ are the \emph{factors} of $\mathbf{L}(z)$ if 
		$\mathbf{L}(z) = \mathbf{L}_{0} \mathbf{B}_{1}(z) \cdots \mathbf{B}_{k}(z)$.
\end{definition}

In \cite{Dzh:2009:FORMF} we showed that left and right divisors can be written explicitly in terms of 
the eigenvectors of $\mathbf{L}(z)$,
	\begin{equation}
		\mathbf{B}^{r}_{s}(z) = \mathbf{I} + \frac{z_{s} - \zeta_{s}}{z - z_{s}} 
		\frac{\mathbf{c}_{s} \mathbf{b}_{s}^{\dag}}{\mathbf{b}_{s}^{\dag} \mathbf{c}_{s}},\qquad
		\mathbf{B}^{l}_{s}(z) = \mathbf{I} + \frac{z_{s} - \zeta_{s}}{z - z_{s}} 
		\frac{\mathbf{a}_{s} \mathbf{d}_{s}^{\dag}}{\mathbf{d}_{s}^{\dag} \mathbf{a}_{s}}, 
		\label{eq:left-right-divs}
	\end{equation}
	and so the coordinate systems described in Lemma~\ref{lem:eigen-coords} are just parameterizations
	of Lax matrices by left (resp.~right) divisors. When $\mathbf{L}(z)$ has only two factors, which is the
	quadratic case that we focus on in this paper, there is no essential difference between divisors and factors.

% subsection elementary_divisors_and_multiplicative_representations (end)

\subsection{Re-factorization transformations and d-$P_{V}$} % (fold)
\label{sub:re_factorization_transformations_and_d_p__v_}
We now consider the refactorization transformation in the quadratic case.  Let 
\begin{equation}
	\mathbf{L}(z) = \mathbf{L}_{0} \mathbf{B}_{1}(z) \mathbf{B}_{2}(z) = \mathbf{B}_{2}^{l}(z) \mathbf{L}_{0} \mathbf{B}_{1}^{r}(z) 
	= \mathbf{B}_{1}^{l}(z) \mathbf{L}_{0} \mathbf{B}_{2}^{r}(z) 
\end{equation}
Then we can consider either isospectral transformation with $\mathbf{R}(z) = \mathbf{B}_{1}^{r}(z)$,
\begin{alignat}{2}
	\mathbf{L}(z) &= \mathbf{B}_{2}^{l}(z) \mathbf{L}_{0} \mathbf{B}^{r}_{1}(z) &\mapsto 
	\tilde{\mathbf{L}}(z)&=\mathbf{B}^{r}_{1}(z) \mathbf{B}^{l}_{2}(z) \mathbf{L}_{0} = 
	\tilde{\mathbf{B}}_{2}^{l}(z) \mathbf{L}_{0} \tilde{\mathbf{B}}_{1}^{r}(z),
	\intertext{or isomonodromic transformations}
	\mathbf{L}(z) &= \mathbf{B}_{2}^{l}(z) \mathbf{L}_{0} \mathbf{B}^{r}_{1}(z) &\mapsto 
	\tilde{\mathbf{L}}(z)&=\mathbf{B}^{r}_{1}(z+1) \mathbf{B}^{l}_{2}(z) \mathbf{L}_{0} = 
	\tilde{\mathbf{B}}_{2}^{l}(z) \mathbf{L}_{0} \tilde{\mathbf{B}}_{1}^{r}(z).
\end{alignat}

When $m=2$, the resulting phase space is two-dimensional and it possible to introduce the so-called \emph{spectral
coordinates} $(\gamma,\pi)$ on the space of Lax matrices. Then, in the isomonodromic case, 
the spectral coordinates of $\mathbf{L}(z)$ and $\tilde{\mathbf{L}}(z)$ are related by the difference 
Painlev\'e V equation in the Sakai's classification \cite{Sak:2001:RSAWARSGPE}, see \cite{Dzh:2009:FORMF}:
\begin{align}
	\tilde{\gamma} + \gamma & = z_{2} + \zeta_{2} + \frac{\rho_{1}(k_{1} - z_{1} + \zeta_{2})}{\pi-\rho_{1}} + 
	\frac{\rho_{2}(k_{2} - z_{1} + \zeta_{2} + 1)}{\pi - \rho_{2}},\\
	\tilde{\pi}\pi & = \rho_{1} \rho_{2}  \frac{(\tilde{\gamma} - \tilde{z}_{2}) (\tilde{\gamma} - \tilde{\zeta}_{2})}{
	(\tilde{\gamma} - \tilde{z}_{1}) (\tilde{\gamma} - \tilde{\zeta}_{1})}.
\end{align} 
% subsection re_factorization_transformations_and_d_p__v_ (end)

In \cite{Dzh:2008:OTLSOTDIAIT} we proved the following
\begin{theorem}\label{thm:Lagrangian}
	The equations of both the isospectral and isomonodromic dynamic can be written in the Lagrangian form with
			\begin{align}
				\mathcal{L}(\mathbf{X},\mathbf{Y},t) & = 
				(z_{2} - z_{1}(t)) \log(\mathbf{x}_{1}^{\dag}  \mathbf{x}_{2}) + 
				(z_{1}(t) - \zeta_{2}) \log(\mathbf{y}_{1}^{\dag} \mathbf{L}_{0}^{-1}\mathbf{x}_{2})\\
				&\qquad +(\zeta_{2} - \zeta_{1}(t))\log(\mathbf{y}_{1}^{\dag} \mathbf{L}_{0}^{-1} \mathbf{y}_{2}) +
				(\zeta_{1}(t) - z_{2}) \log(\mathbf{x}_{1}^{\dag} \mathbf{y}_{2}),			
			\end{align}	
			where $\mathbf{X}=(\mathbf{x}_{1},\mathbf{x}_{2}^{\dag})$ and $\mathbf{Y}=(\mathbf{y}_{1},\mathbf{y}_{2}^{\dag})$,
			in the isomonodromic case $z_{1}(t) = z_{1} - t$, $\zeta_{1}(t) = \zeta_{1} - t$, and in the isospectral case
			$z_{1}(t) = z_{1}$, $\zeta_{1}(t)=\zeta_{1}$ and $\mathcal{L}(X,Y)$ is time-independent.
\end{theorem}
Here $\mathbf{x}_{i}^{\dag}$, $\mathbf{x}_{j}$, $\mathbf{y}_{i}^{\dag}$, $\mathbf{y}_{j}$ are some eigenvectors of 
$\mathbf{L}(z)$ and $\tilde{\mathbf{L}}(z)$.

The main goal of the present paper is to give a combinatorial--geometric representation of this generating function.
% section preliminaries (end)

\section{Linear Equations given by Elementary Divisors and their Geometric Representation} % (fold)
\label{sec:geometric_representation_of_some_linear_equations}

In studying the relationship between various eigenvectors of $\mathbf{L}(z)$ we mainly work with 
linear equations of the form $\mathbf{v} \sim \mathbf{B}(z^{*}) \mathbf{w}$ and 
$\mathbf{v}^{\dag} \sim \mathbf{w}^{\dag}\mathbf{B}(z^{*})$, where $\sim$ means that the two vectors are 
proportional and $\mathbf{B}(z^{*}; z_{0},\zeta_{0})$ is a matrix of the form~(\ref{eq:el-div}) evaluated at some point
$z^{*}$. In this section we give a diagrammatic representation of such equations using the language of
\emph{elementary triples}. 

\subsection{Basic Definitions} % (fold)
\label{sub:basic_definitions}
In what follows, let $\mathbf{V}$ be an $n$-dimensional complex vector space whose elements are 
column vectors, $\mathbf{V}\simeq \mathbb{C}^{n}$, and let $\mathbf{V}^{\dag}$ be its dual-space.
We think of elements of $\mathbf{V}^{\dag}$ as row-vectors and denote its elements 
as $\mathbf{q}^{\dag}\in \mathbf{V}^{\dag}\simeq (\mathbb{C}^{n})^{\dag}$. Our main definition is the following.

\begin{definition}\label{def:v-triple}
	Let $\lambda_{i}\in \mathbb{P}\mathbf{V}$, $i=1,\dots 3$, be one-dimensional linear subspaces of $\mathbf{V}$ and
	$\mu^{\dag}\in \mathbb{P}\mathbf{V}^{\dag}$ be a one-dimensional linear subspace of $\mathbf{V}^{\dag}$. We say that
	$(\lambda_{1},\lambda_{2},\lambda_{3})$ form a \emph{$\mu^{\dag}$-based triple}  with parameters
	$(\alpha_{1},\alpha_{2},\alpha_{3})$, $\alpha_{i}\in \mathbb{C}\backslash \{0\}$, %\
	if the following conditions hold:

	\begin{itemize}
		\item $\lambda_{i}\notin \ker \mu^{\dag}$ for all $i$;
		\item $\lambda_{i}\neq \lambda_{j}$ for $i\neq j$ but $\lambda_{i}\subset \lambda_{j} + \lambda_{k}$ for all $i, j, k$.
		\item $\displaystyle
			\sum_{i=1}^{3}\frac{\alpha_{i}}{\mu^{\dag} \lambda_{i}} \lambda_{i} = 
			\frac{\alpha_{1}}{\mu^{\dag} \lambda_{1}} \lambda_{1} + \frac{\alpha_{2}}{\mu^{\dag} \lambda_{2}} \lambda_{2} + 
			\frac{\alpha_{3}}{\mu^{\dag} \lambda_{3}} \lambda_{3} = 0.$
	\end{itemize}		
	
	% \begin{center}
	% 	\begin{tabular}{cc}
	% 		\begin{minipage}{4in}
	% 			\begin{itemize}
	% 				\item $\lambda_{i}\notin \ker \mu^{\dag}$ for all $i$;
	% 				\item $\lambda_{i}\neq \lambda_{j}$ for $i\neq j$;
	% 				\item $\lambda_{i}\subset \lambda_{j} + \lambda_{k}$ for all $i, j, k$;
	% 				\item $\displaystyle
	% 					\sum_{i=1}^{3}\frac{\alpha_{i}}{\mu^{\dag} \lambda_{i}} \lambda_{i} = 
	% 					\frac{\alpha_{1}}{\mu^{\dag} \lambda_{1}} \lambda_{1} + \frac{\alpha_{2}}{\mu^{\dag} \lambda_{2}} \lambda_{2} + 
	% 					\frac{\alpha_{3}}{\mu^{\dag} \lambda_{3}} \lambda_{3} = 0.$
	% 			\end{itemize}		
	% 		\end{minipage}	& 
	% 			\raisebox{-0.5in}{\includegraphics[height=1.25in]{c-gr-1}}		
	% 	\end{tabular}
	% \end{center}	
	The last equation above adjusts the ``slope'' of $\lambda_{i}$ w.r.t.~$\lambda_{j}$ and $\lambda_{k}$ and it should be 
	interpreted in terms of spanning vectors. 
	To that end, let $\lambda_{i} = \operatorname{Span}_{\mathbb{C}}\{\mathbf{p}_{i}\}$ 
	and $\mu^{\dag} = \operatorname{Span}_{\mathbb{C}}\{\mathbf{q}^{\dag}\}$. 
	Then we require that the vectors $\mathbf{p}_{1}$, $\mathbf{p}_{2}$, and 
	$\mathbf{p}_{3}$ are linearly dependent but pairwise independent, $\mathbf{q}^{\dag} \mathbf{p}_{i}\neq 0$ for all $i$,
	and the last equation takes the form
	\begin{equation}\label{eq:p-dep}
			\sum_{i=1}^{3}\frac{\alpha_{i}}{\mathbf{q}^{\dag} \mathbf{p}_{i}} \mathbf{p}_{i} = 
			\frac{\alpha_{1}}{\mathbf{q}^{\dag} \mathbf{p}_{1}} \mathbf{p}_{1} + 
			\frac{\alpha_{2}}{\mathbf{q}^{\dag} \mathbf{p}_{2}} \mathbf{p}_{2} + 
			\frac{\alpha_{3}}{\mathbf{q}^{\dag} \mathbf{p}_{3}} \mathbf{p}_{3} = \mathbf{0}.
			% \\\text{where}\qquad c_{1} \mathbf{p}_{1} + c_{2} \mathbf{p}_{2} + c_{3} \mathbf{p}_{3} &= \mathbf{0}.
			% \qquad\qquad
			% 	\raisebox{-0.7in}{\includegraphics[height=1.2in]{c-gr-2}}
	\end{equation}
	Note that this equation is homogeneous and so is independent of the choice of basis vectors.	
	We represent such triples by diagrams of the form
	\begin{center}
		\includegraphics[height=1.2in]{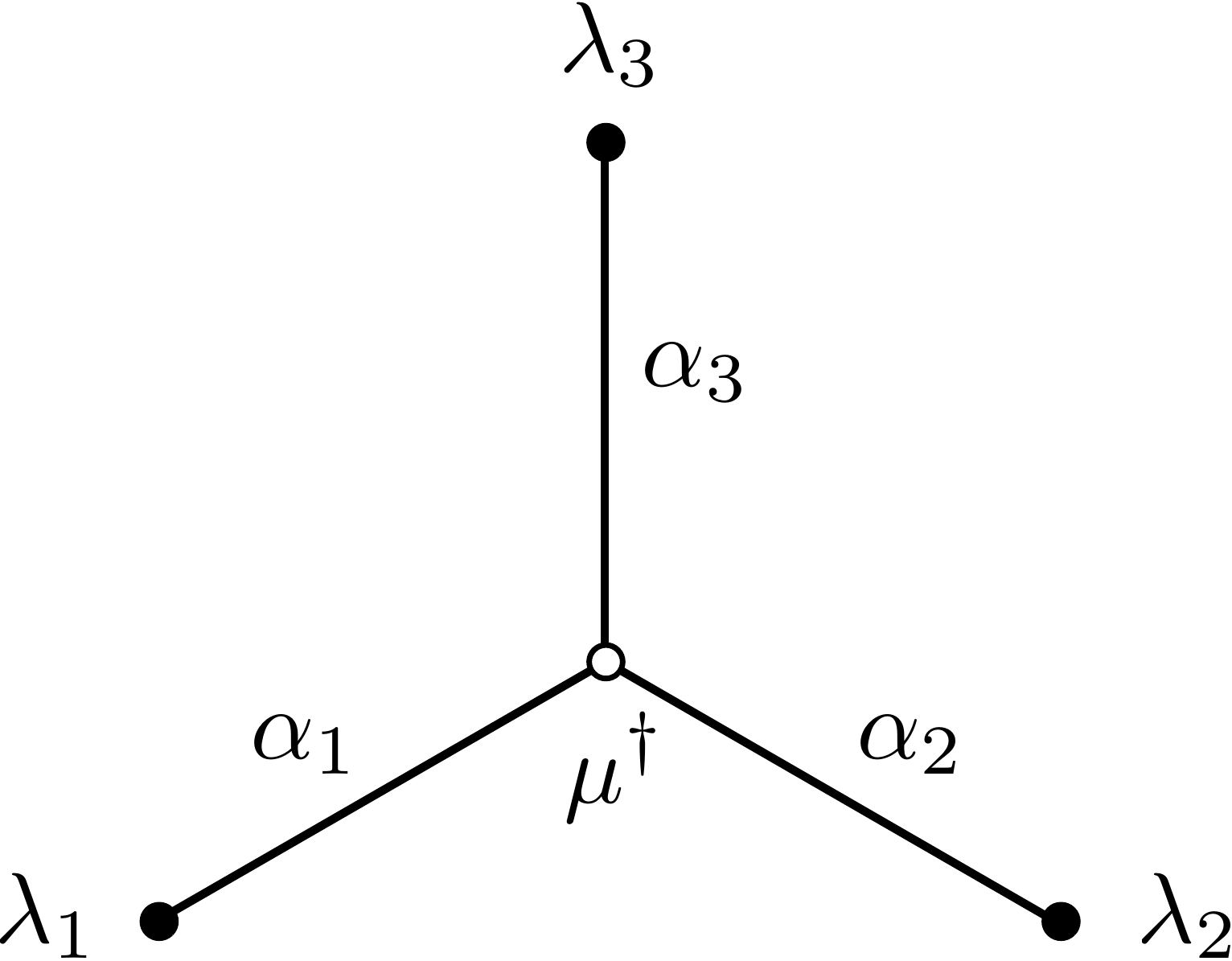} \qquad 		\includegraphics[height=1.2in]{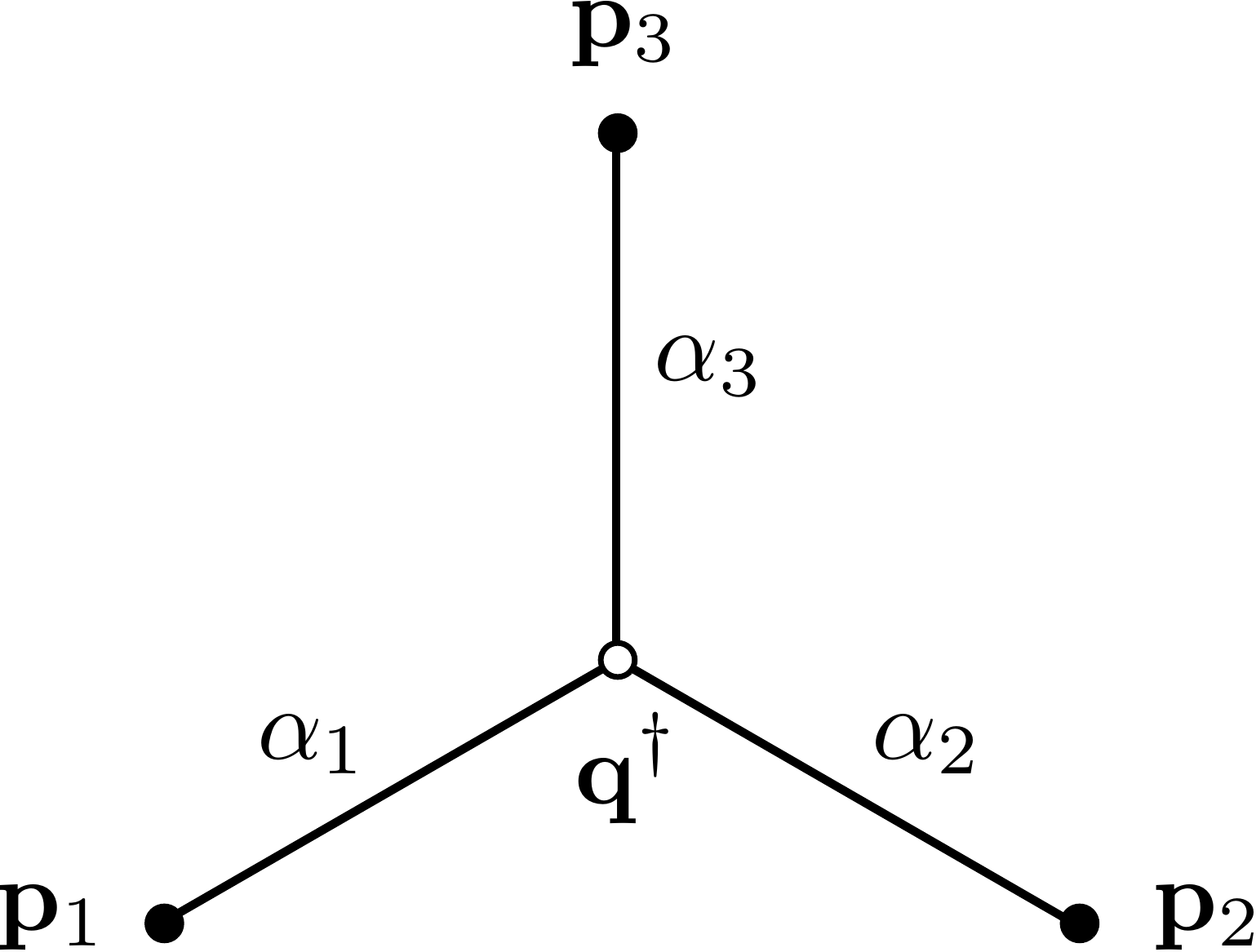},
	\end{center}
	where we use the black circles for 
	subspaces of $\mathbf{V}$ (or their basis vectors) and white circles for subspaces of $\mathbf{V}^\dag$. 	
	In what follows, we switch freely between subspace and vector formulations.
\end{definition}

It is clear that $\alpha_{i}$s are defined up to a common multiplicative constant and satisfy
\begin{equation}
	\alpha_{1} + \alpha_{2} + \alpha_{3} = 0.
\end{equation}
Moreover, given $\mathbf{q}^{\dag}$, there is a \emph{unique} vector 
$[\mathbf{p}_{i}] = [\mathbf{p}_{i}]_{\mathbf{q}^{\dag}}\in \lambda_{i}$ 
normalized by the condition $\mathbf{q}^{\dag} [\mathbf{p}]=1$, i.e.,
\begin{equation}
	[\mathbf{p}_{i}] = \frac{\mathbf{p}_{i}}{\mathbf{q}^{\dag} \mathbf{p}_{i}}\in \lambda_{i},
\end{equation}
where $\mathbf{p}_{i}$ is \emph{any} non-zero vector in $\lambda_{i}$. Then, in terms of the normalized vectors, the 
linear dependence equation (\ref{eq:p-dep}) becomes
\begin{equation}
	\alpha_{1} [\mathbf{p}_{1}] + 	\alpha_{2} [\mathbf{p}_{2}] + 	\alpha_{3} [\mathbf{p}_{3}] = \mathbf{0}\qquad \text{or}\qquad
	[\mathbf{p}_{i}] = \frac{\alpha_{j} [\mathbf{p}_{j}] + \alpha_{k} [\mathbf{p}_{k}]}{\alpha_{j} + \alpha_{k}},\quad
	i,j,k\text{ all distinct},
\end{equation}
% where in the last equation all indices are distinct, 
which explains the interpretation of this equation in terms of ``slopes''. 
% Further, the vector of 
% coefficient of linear dependence of $\mathbf{p}_{i}$ then should satisfy
% \begin{equation}
% 	\langle c_{1},c_{2},c_{3} \rangle \sim 
% 	\langle \frac{\alpha_{1}}{\mathbf{q}^{\dag} \mathbf{p}_{1}},\frac{\alpha_{2}}{\mathbf{q}^{\dag} \mathbf{p}_{2}},
% 	\frac{\alpha_{3}}{\mathbf{q}^{\dag} \mathbf{p}_{3}} \rangle\qquad\text{or}\qquad
% 	\frac{c_{1}\mathbf{q}^{\dag} \mathbf{p}_{1}}{\alpha_{1}} = 
% 	\frac{c_{2}\mathbf{q}^{\dag} \mathbf{p}_{2}}{\alpha_{2}} =
% 	\frac{c_{3}\mathbf{q}^{\dag} \mathbf{p}_{3}}{\alpha_{3}}.
% \end{equation}
% 
% Conversely, given a linear-dependent triple $c_{1} \mathbf{p}_{1} + c_{2} \mathbf{p}_{2} + c_{3} \mathbf{p}_{3} = \mathbf{0}$,
% $c_{i}\neq0$ for all $i$,
% and a vector of parameters $(\alpha_{1},\alpha_{2},\alpha_{3})$, $\sum_{i} \alpha_{i}=0$, it is possible to 
% find a co-vector $\mathbf{q}^{\dag}$ (in a non-unique way) to make it into a $\mathbf{q}^{\dag}$ triple. 
% Indeed,
% \begin{equation}
% 	\mathbf{q}^{\dag} = \frac{\alpha_{1}}{c_{1}}\mathbf{q}_{1}^{\dag} + \frac{\alpha_{2}}{c_{2}} \mathbf{q}_{2}^{\dag}
% 	+ \frac{\alpha_{3}}{c_{3}} \mathbf{q}_{3}^{\dag} + \mathbf{Q}^{\dag},\qquad
% 	\text{where } \mathbf{q}_{i}^{\dag} \mathbf{p}_{j} = \delta_{ij},\quad \mathbf{Q}^{\dag} \mathbf{p}_{i}=0\quad
% 	\forall i,j.
% \end{equation}
% subsection triples (end)

\begin{definition}\label{def:w-triple}
We define the \emph{$\lambda$-based dual triples} 
$(\mu_{1},\mu_{2},\mu_{3})$ with parameters 
$(\beta_{1},\beta_{2},\beta_{3})$, $\beta_{i}\in \mathbb{C}\backslash \{0\}$, where
$\lambda\in \mathbb{P}\mathbf{V}$ is a one-dimensional linear subspace of $\mathbf{V}$ and
$\mu_{i}^{\dag}\in \mathbb{P}\mathbf{V}^{\dag}$, $i=1,\dots 3$, are one-dimensional linear subspaces of
$\mathbf{V}^{\dag}$  in exactly the same way, and represent them by the following diagrams:

\begin{center}
	\includegraphics[height=1.2in]{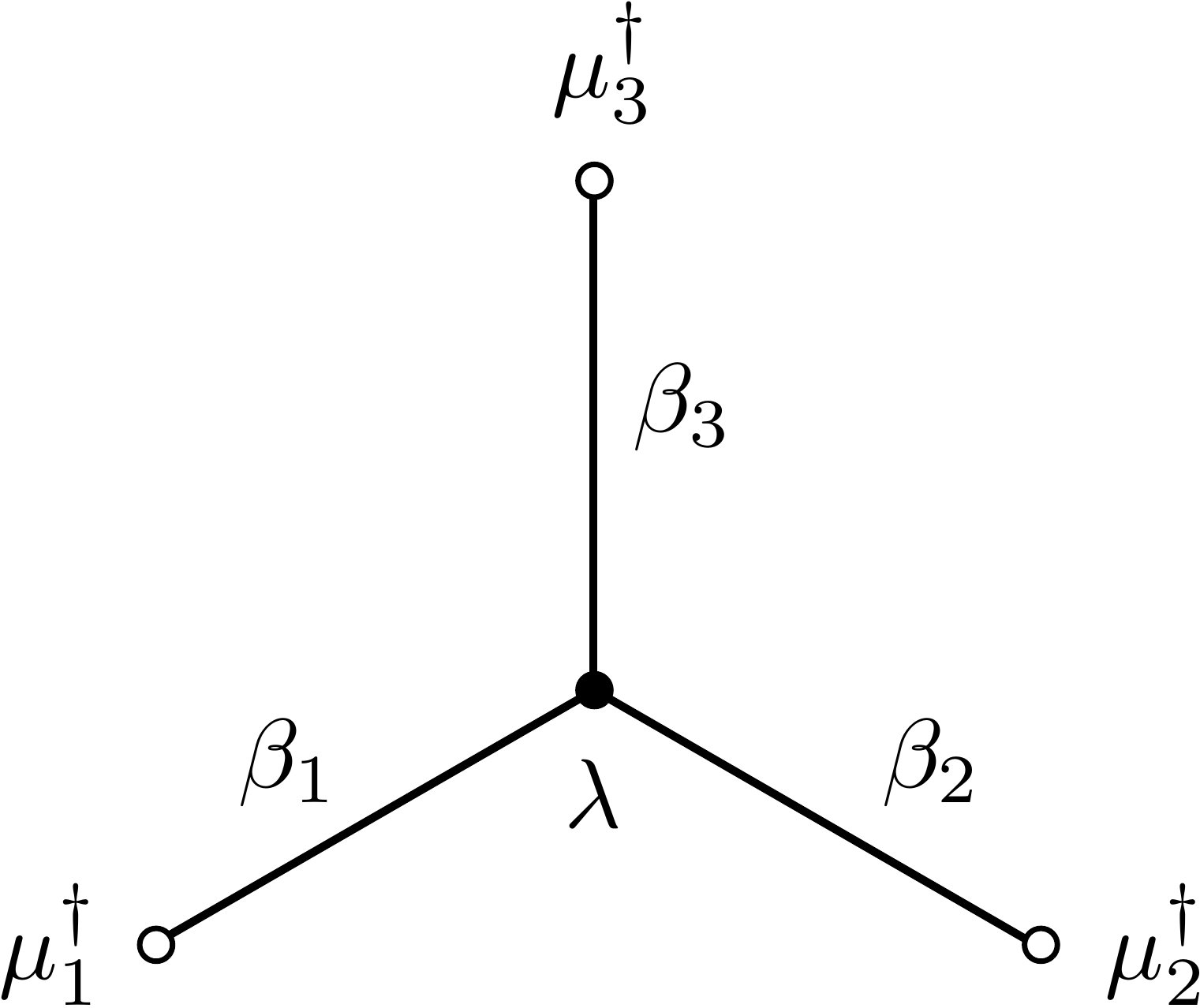} \qquad\qquad\qquad \includegraphics[height=1.2in]{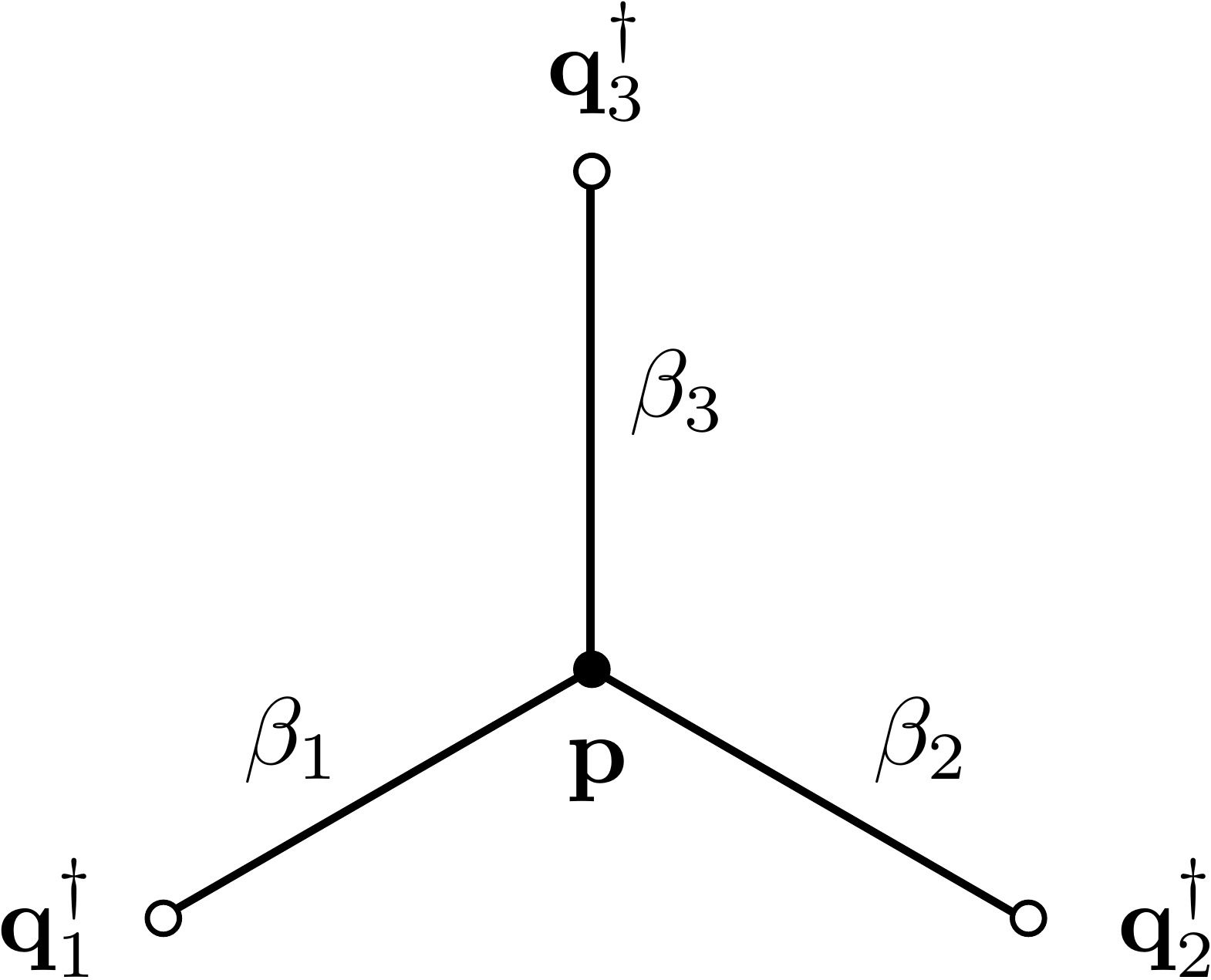}\quad.
\end{center}
\end{definition}

Finally, the behavior of the triples under the action of $\mathbf{GL}(\mathbf{V})$
is described by the following obvious Lemma.
\begin{lemma}\label{lem:group action}
	Let $\mathbf{A}\in \mathbf{GL}(\mathbf{V})$.
	\begin{enumerate}[(i)]
		\item If $(\lambda_{1},\lambda_{2},\lambda_{3})$
		is a $\mu^{\dag}$-based triple with parameters $(\alpha_{1},\alpha_{2},\alpha_{3})$, then	
		$(\mathbf{A}\lambda_{1}, \mathbf{A}\lambda_{2}, \mathbf{A}\lambda_{3})$ is a 
		$\mu^{\dag} \mathbf{A}^{-1}$-based triple with the same parameters $(\alpha_{1},\alpha_{2},\alpha_{3})$;
		\item dually, if $(\mu_{1}^{\dag},\mu_{2}^{\dag},\mu_{3}^{\dag})$
		be a $\lambda$-based triple with parameters $(\beta_{1},\beta_{2},\beta_{3})$, then 
		$(\mu_{1}^{\dag}\mathbf{A}^{-1},\mu_{2}^{\dag}\mathbf{A}^{-1},\mu_{3}^{\dag}\mathbf{A}^{-1})$ is an 
		$\mathbf{A}\lambda$-based triple with the same parameters $(\beta_{1},\beta_{2},\beta_{3})$.
	\end{enumerate}
\end{lemma}

% subsection basic_definitions (end)

\subsection{Elementary Divisors and Triples} % (fold)
\label{sub:elementary_divisors_and_triples}
We now turn our attention to the study of the relationship between triples and elementary divisors. 
The following Lemma is immediate.

\begin{lemma}\label{lem:el-div-triples}
	Let 
	\begin{equation}
	\mathbf{B} = \mathbf{I} + \frac{ z_{0} - \zeta_{0} }{ z^{*} - z_{0} } \frac{\mathbf{p} \mathbf{q}^{\dag}}{\mathbf{q}^{\dag} \mathbf{p}}
	= \mathbf{I} + \frac{ z_{0} - \zeta_{0} }{ z^{*} - z_{0} } [\mathbf{p}] \mathbf{q}^{\dag} = 
	\mathbf{I} + \frac{ z_{0} - \zeta_{0} }{ z^{*} - z_{0} } \mathbf{p} [\mathbf{q}^{\dag}],	
	\end{equation}
	$[\mathbf{p}] = [\mathbf{p}]_{\mathbf{q}^{\dag}}$ and $[\mathbf{q}^{\dag}] = [\mathbf{q}^{\dag}]_{\mathbf{p}}$. 
	Note that we can normalize either $\mathbf{p}$ or $\mathbf{q}^{\dag}$, but not both. The choice of which one 
	should be normalized depends on whether we think of $\mathbf{B}$ as acting on the 
	elements of $\mathbf{V}$ or $\mathbf{V}^{\dag}$. 
	Then
	\begin{enumerate}[(i)]
		\item Vectors $\mathbf{p}$ and $\mathbf{q}^{\dag}$ are the eigenvectors of $\mathbf{B}$,
		\begin{equation}
			\mathbf{q}^{\dag} \mathbf{B} = \frac{ z^{*} - \zeta_{0} }{ z^{*} - z_{0} } \mathbf{q}^{\dag},\qquad 
			\mathbf{B}\mathbf{p} = \frac{ z^{*} - \zeta_{0} }{ z^{*} - z_{0} }\mathbf{p},\qquad\text{where}\quad
			\det \mathbf{B} = \frac{ z^{*} - \zeta_{0} }{ z^{*} - z_{0} }.
		\end{equation}
		
		\item For any vector $\mathbf{w}$ such that $\mathbf{q}^{\dag} \mathbf{w} \neq 0$, and for any 
		row-vector $\mathbf{w}^{\dag}$ such that $\mathbf{w}^{\dag} \mathbf{p} \neq 0$, using normalized vectors
		we get 
		\begin{equation}
 			(\mathbf{B} - \mathbf{I}) [\mathbf{w}]_{\mathbf{q}^{\dag}} = 
			\frac{ z_{0} - \zeta_{0} }{ z^{*} - z_{0} } [\mathbf{p}]_{\mathbf{q}^{\dag}},\qquad 
			[\mathbf{w}^{\dag}]_{\mathbf{p}}(\mathbf{B} - \mathbf{I}) = \frac{ z_{0} - \zeta_{0} }{ z^{*} - z_{0} } [\mathbf{q}^{\dag}]_\mathbf{p}.
		\end{equation}
		\item If $\mathbf{v} = \mathbf{B} \mathbf{w}$, then 
		$[\mathbf{v}]_{\mathbf{q}^{\dag}} = (z^{*} - z_{0})/(z^{*} - \zeta_{0}) \mathbf{B}[\mathbf{w}]_{\mathbf{q}^{\dag}}$, and so 
		$(\mathbf{w},\mathbf{p},\mathbf{v})$ form a $\mathbf{q}^{\dag}$-triple with parameters
		$(z^{*} - z_{0},z_{0}  - \zeta_{0},\zeta_{0} - z^{*})$. Similarly, if $\mathbf{v}^{\dag} = \mathbf{w}^{\dag}\mathbf{B}$,
		then $[\mathbf{v}^{\dag}]_{\mathbf{p}} = (z^{*} - z_{0})/(z^{*} - \zeta_{0})[\mathbf{w}^{\dag}]_{\mathbf{p}} \mathbf{B}$,
		and $(\mathbf{w}^{\dag},\mathbf{p}^{\dag},\mathbf{v}^{\dag})$ form a $\mathbf{p}$-triple with parameters
		$(z^{*} - z_{0},z_{0}  - \zeta_{0},\zeta_{0} - z^{*})$.
		\item Conversely, for any $\mathbf{q}^{\dag}$-triple $(\mathbf{p}_{1}, \mathbf{p}_{2}, \mathbf{p}_{3})$ 
		with parameters $(\alpha_{1}, \alpha_{2}, \alpha_{3})$,
		\begin{equation}
			[\mathbf{p}_{3}] = \frac{\alpha_{1}}{\alpha_{1} + \alpha_{2}}
			\left(\mathbf{I} + \frac{\alpha_{2}}{\alpha_{1}}[\mathbf{p}_{2}] \mathbf{q}^{\dag}\right) [\mathbf{p}_{1}]
			= \frac{1}{\alpha_{1} + \alpha_{2}} \frac{\partial}{\partial \mathbf{q}^{\dag}}
			\left(\alpha_{1} \log(\mathbf{q}^{\dag} \mathbf{p}_{1}) + \alpha_{2} \log(\mathbf{q}^{\dag}\mathbf{p}_{2}) \right),
		\end{equation}
		and for any $\mathbf{p}$-triple $(\mathbf{q}_{1}^{\dag},\mathbf{q}_{2}^{\dag},\mathbf{q}_{3}^{\dag})$
		with parameters $(\beta_{1},\beta_{2},\beta_{3})$,
		\begin{equation}
			[\mathbf{q}_{3}^{\dag}] = \frac{\beta_{1}}{\beta_{1} + \beta_{2}}
			[\mathbf{q}_{1}^{\dag}]\left(\mathbf{I} + \frac{\beta_{2}}{\beta_{1}}\mathbf{p}[\mathbf{q}_{2}^{\dag}] \right) 
			= \frac{1}{\beta_{1} + \beta_{2}} \frac{\partial}{\partial \mathbf{p}}
			\left(\beta_{1} \log(\mathbf{q}^{\dag}_{1} \mathbf{p}) + \beta_{2} \log(\mathbf{q}^{\dag}_{2}\mathbf{p}) \right).
		\end{equation}
	\end{enumerate}
	
\end{lemma}
% subsection elementary_divisors_and_triples (end)	

\subsection{Gluing Properties} % (fold)
\label{sub:gluing_properties}
The following two Lemmas  show that 
any two triples with matching parts can be ``glued'' together and then completed
into a \emph{cube} in such a way that any vertex of the cube makes a triple. More precisely,
the \emph{Face Lemma} shows how to glue two triples  of the same type that share two vertices,
and the \emph{Edge Lemma} shows how to glue two triples of different type if the center of one is the end of the 
other and vise-versa, and the common edge parameters coincide.

\begin{lemma}[The \emph{Face} Lemma]\label{lem:face-lemma}
	Let the pairs $\mathbf{p}_{1}, \mathbf{p}_{2}\in \mathbf{V}$ and $\mathbf{q}_{1}^{\dag}, \mathbf{q}_{2}^{\dag}\in \mathbf{V}^{\dag}$
	be linearly independent and not-orthogonal, $\mathbf{q}_{i}^{\dag} \mathbf{p}_{j}\neq0$, and let 
	$\alpha_{1}, \alpha_{2}, \beta_{1},\beta_{2}$  be any non-zero complex parameters. Then the quad 
	\begin{center}
		\includegraphics[width=1in]{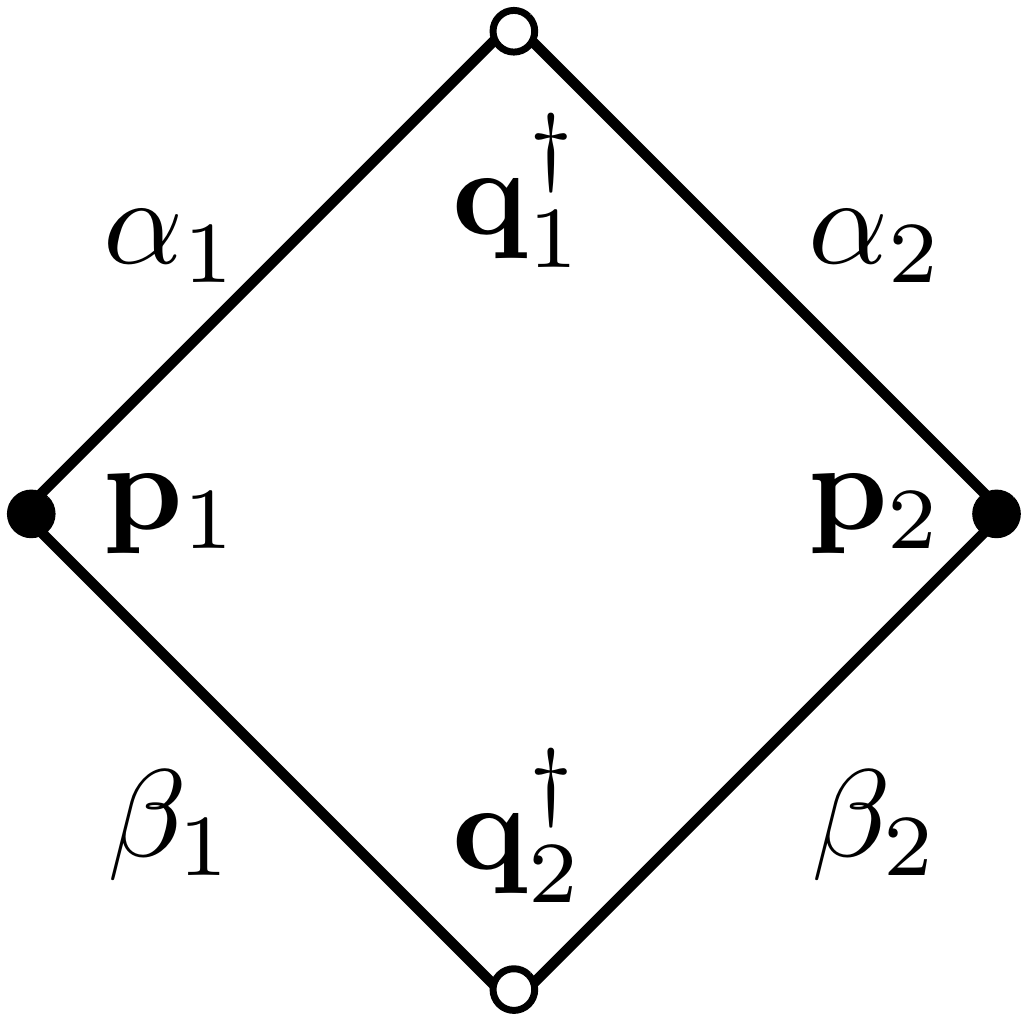}
	\end{center} 
	can be completed to a cube with 
	\begin{equation}
		\alpha_{3} = - (\alpha_{1} + \alpha_{2}),\quad \beta_{3} = -(\beta_{1} + \beta_{2}),\quad 
		\gamma_{1} = -(\alpha_{1} + \beta_{1}), \quad \gamma_{2} = -(\alpha_{2} + \beta_{2}),
	\end{equation}
	and some other parameters $\alpha_{4}, \alpha_{5}, \beta_{4}, \beta_{5}$,
	\begin{center}
		\includegraphics[height=1.5in]{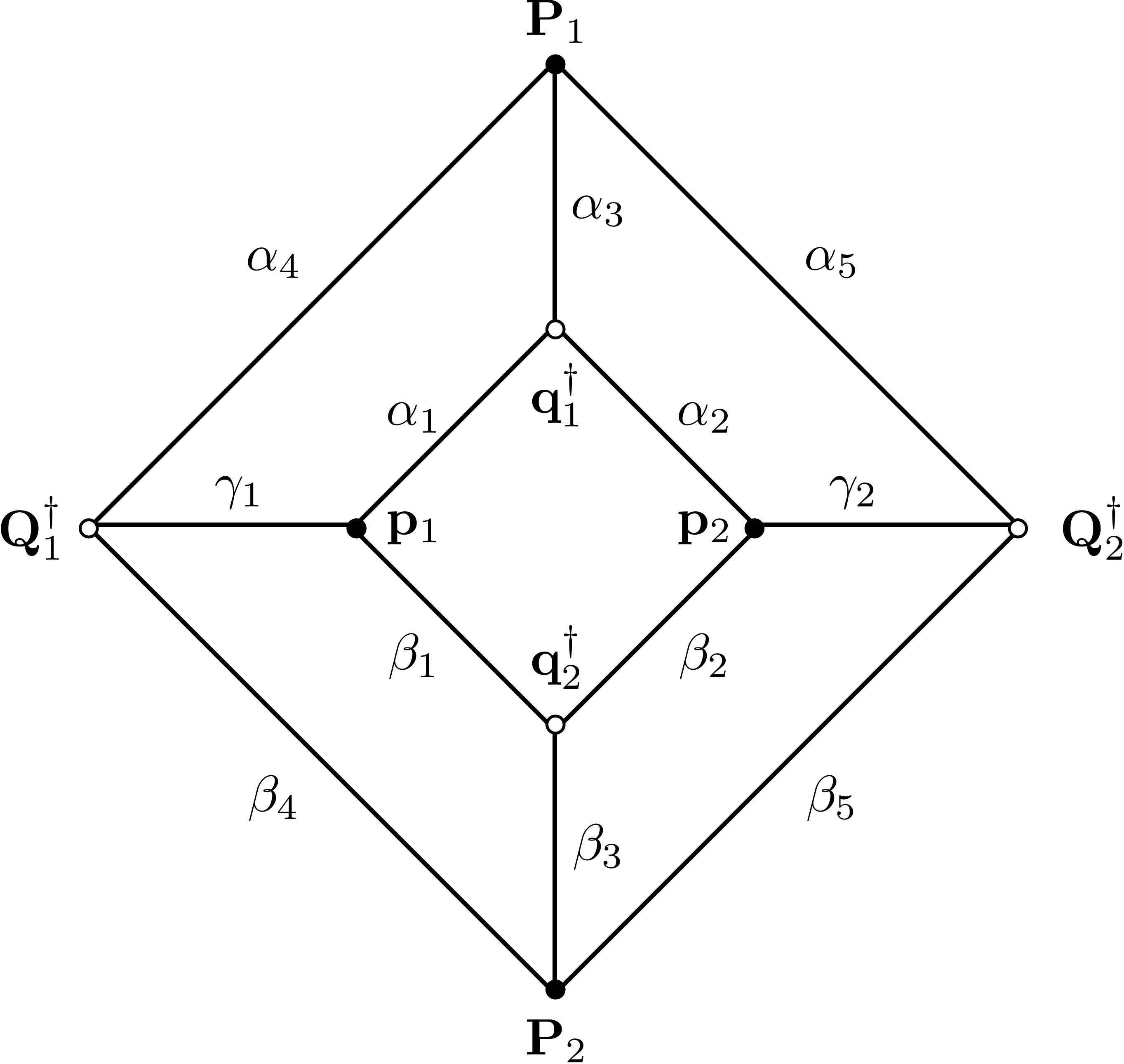}\qquad \includegraphics[height=1.5in]{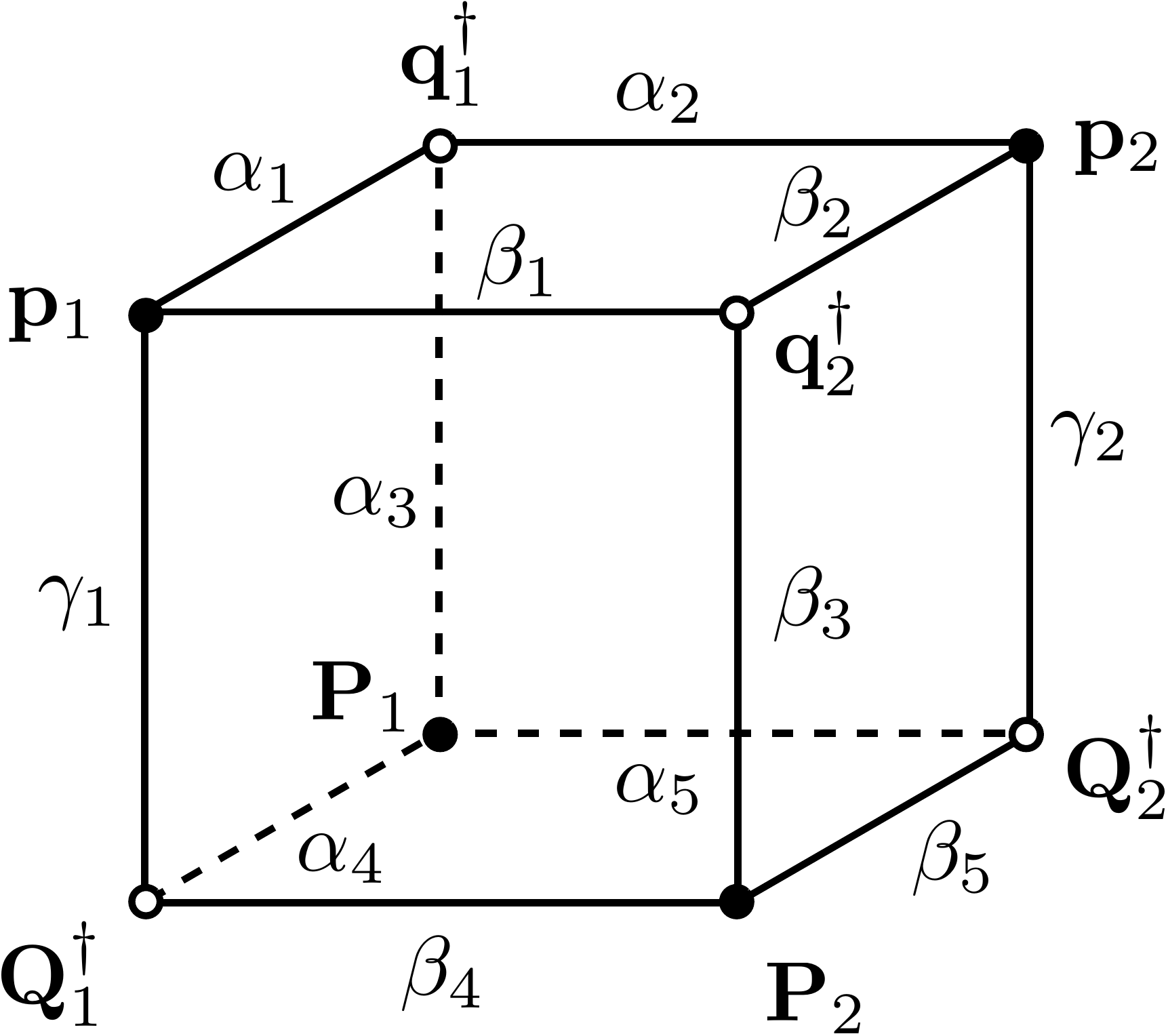}
	\end{center}
	so that at each vertex the corresponding triple condition is satisfied. Moreover, if the parameters satisfy the additional
	cyclic condition
	\begin{equation}
		\alpha_{1} + \alpha_{2} + \beta_{1} + \beta_{2} = 0,
	\end{equation}
	on the initial face, then the same condition is satisfied on the remaining five faces, and so 
	\begin{equation}
		\alpha_{4} = - \beta_{2},\quad \alpha_{5} = - \beta_{1},\quad \beta_{4} = - \alpha_{2}, \quad \beta_{5} = - \alpha_{1}.
	\end{equation}
	Further, let 
	\begin{equation}
		\mathcal{L} = \alpha_{1} \log(\mathbf{q}_{1}^{\dag} \mathbf{p}_{1}) + \alpha_{2} \log(\mathbf{q}_{1}^{\dag} \mathbf{p}_{2})
		 + \beta_{1} \log(\mathbf{q}_{2}^{\dag} \mathbf{p}_{1}) + \beta_{2}\log(\mathbf{q}_{1}^{\dag} \mathbf{p}_{2}).
	\end{equation}
	Then
	\begin{equation}
		[\mathbf{P}_{1}]_{\mathbf{q}_{1}^{\dag}} = \alpha_{3} \frac{\partial \mathcal{L}}{\partial \mathbf{q}_{1}^{\dag}},\qquad
		[\mathbf{P}_{2}]_{\mathbf{q}_{2}^{\dag}} = \beta_{3} \frac{\partial \mathcal{L}}{\partial \mathbf{q}_{2}^{\dag}},\qquad
		[\mathbf{Q}_{1}^{\dag}]_{\mathbf{p}_{1}} = \gamma_{1} \frac{\partial \mathcal{L}}{\partial \mathbf{p}_{1}},\qquad		
		[\mathbf{Q}_{2}^{\dag}]_{\mathbf{p}_{2}} = \gamma_{2} \frac{\partial \mathcal{L}}{\partial \mathbf{p}_{1}}.
	\end{equation}	
\end{lemma}
\begin{proof}
	First, note that the normalized outside vertices are uniquely determined by the given data,
	\begin{align}
		[\mathbf{P}_{1}]_{\mathbf{q}_{1}^{\dag}} &= \frac{ \alpha_{1} [\mathbf{p}_{1}]_{\mathbf{q}_{1}^{\dag}} + 
		\alpha_{2} [\mathbf{p}_{2}]_{\mathbf{q}_{1}^{\dag}}}{\alpha_{1} + \alpha_{2}} &
		[\mathbf{P}_{2}]_{\mathbf{q}_{2}^{\dag}} &= \frac{ \beta_{1} [\mathbf{p}_{1}]_{\mathbf{q}_{2}^{\dag}} + 
		\beta_{2} [\mathbf{p}_{2}]_{\mathbf{q}_{2}^{\dag}}}{\beta_{1} + \beta_{2}} \\
		[\mathbf{Q}_{1}^{\dag}]_{\mathbf{p}_{1}} &= \frac{ \alpha_{1} [\mathbf{q}_{1}^{\dag}]_{\mathbf{p}_{1}} + 
		\beta_{1} [\mathbf{q}_{2}^{\dag}]_{\mathbf{p}_{1}}}{\alpha_{1} + \beta_{1}} & 
		[\mathbf{Q}_{2}^{\dag}]_{\mathbf{p}_{2}} &= \frac{ \alpha_{2} [\mathbf{q}_{1}^{\dag}]_{\mathbf{p}_{2}} + 
		\beta_{2} [\mathbf{q}_{2}^{\dag}]_{\mathbf{p}_{2}}}{\alpha_{2} + \beta_{2}}.  
	\end{align}
	Thus, to prove the theorem we need to show that it is possible to find the parameters $\alpha_{4}, \alpha_{5}, \beta_{4}, \beta_{5}$
	so that 
	\begin{equation}
		\alpha_{4} + \alpha_{5} = \alpha_{1} + \alpha_{2}, \quad
		\beta_{4} + \beta_{5} = \beta_{1} + \beta_{2}, \quad
		\alpha_{4} + \beta_{4} = \alpha_{1} + \beta_{1},\quad 
		\alpha_{5} + \beta_{5} = \alpha_{2} + \beta_{2},
	\end{equation}
	and that the following equations hold at the outside nodes:
	\begin{alignat}{2}
		\mathbf{Q}_{1}^{\dag}&: &\qquad \alpha_{4}[\mathbf{P}_{1}]_{\mathbf{Q}_{1}^{\dag}} + 
		\beta_{4}[\mathbf{P}_{2}]_{\mathbf{Q}_{1}^{\dag}} + \gamma_{1} [\mathbf{p}_{1}]_{\mathbf{Q}_{1}^{\dag}} &= \mathbf{0}\\
		\mathbf{Q}_{2}^{\dag}&: &\qquad \alpha_{5}[\mathbf{P}_{1}]_{\mathbf{Q}_{2}^{\dag}} + 
		\beta_{5}[\mathbf{P}_{2}]_{\mathbf{Q}_{2}^{\dag}} + \gamma_{2} [\mathbf{p}_{2}]_{\mathbf{Q}_{2}^{\dag}} &= \mathbf{0}\\
		\mathbf{P}_{1}&: &\qquad \alpha_{4}[\mathbf{Q}_{1}^{\dag}]_{\mathbf{P}_{1}} + 
		\alpha_{5}[\mathbf{Q}_{2}^{\dag}]_{\mathbf{P}_{1}} + \alpha_{3} [\mathbf{q}_{1}^{\dag}]_{\mathbf{P}_{1}} &= \mathbf{0}\\		
		\mathbf{P}_{2}&: &\qquad \beta_{4}[\mathbf{Q}_{1}^{\dag}]_{\mathbf{P}_{2}} + 
		\beta_{5}[\mathbf{Q}_{2}^{\dag}]_{\mathbf{P}_{2}} + \beta_{3} [\mathbf{q}_{2}^{\dag}]_{\mathbf{P}_{2}} &= \mathbf{0}.
	\end{alignat}
	Using the change of normalization rule 
	\begin{equation}
	[\mathbf{p}]_{\mathbf{q}_{i}} = \frac{\mathbf{q}_{j}^{\dag}\mathbf{p}}{\mathbf{q}_{i}^{\dag}\mathbf{p}}
	[\mathbf{p}]_{\mathbf{q}_{j}^{\dag}},	
	\end{equation}
	the first equation, when decomposed in the spanning basis $\mathbf{p}_{1}, \mathbf{p}_{2}$, becomes the following system:
	\begin{equation}
		\left\{\begin{aligned}
			\left(\frac{\alpha_{4}}{\alpha_{1} + \alpha_{2}}
			\frac{\mathbf{q}_{1}^{\dag} \mathbf{P}_{1}}{\mathbf{Q}_{1}^{\dag}\mathbf{P}_{1}}\right) \cdot
			\alpha_{1}\frac{\mathbf{Q}_{1}^{\dag}\mathbf{p}_{1}}{\mathbf{q}_{1}^{\dag} \mathbf{p}_{1}} + 
			\left(\frac{\beta_{4}}{ \beta_{1} + \beta_{2}}
			\frac{\mathbf{q}_{2}^{\dag} \mathbf{P}_{2}}{\mathbf{Q}_{1}^{\dag} \mathbf{P}_{2}}\right)\cdot 
			\beta_{1}\frac{\mathbf{Q}_{1}^{\dag} \mathbf{p}_{1}}{\mathbf{q}_{2}^{\dag} \mathbf{p}_{1}} &= \alpha_{1} + \beta_{1} \\
			\left(\frac{\alpha_{4}}{\alpha_{1} + \alpha_{2}}
			\frac{\mathbf{q}_{1}^{\dag} \mathbf{P}_{1}}{\mathbf{Q}_{1}^{\dag}\mathbf{P}_{1}}\right) \cdot
			\alpha_{2}\frac{\mathbf{q}_{2}^{\dag}\mathbf{p}_{2}}{\mathbf{q}_{1}^{\dag} \mathbf{p}_{2}} + 
			\left(\frac{\beta_{4}}{ \beta_{1} + \beta_{2}}
			\frac{\mathbf{q}_{2}^{\dag} \mathbf{P}_{2}}{\mathbf{Q}_{1}^{\dag} \mathbf{P}_{2}}\right)\cdot 
			\beta_{2} &= 0
		\end{aligned}\right.
	\end{equation} 
	Solving it gives
	\begin{align}
		\alpha_{4} &= \beta_{2} (\alpha_{1} + \alpha_{2})(\alpha_{1} + \beta_{1})
		[\mathbf{Q}_{1}^{\dag}]_{\mathbf{p}_{1}}[\mathbf{P}_{1}]_{\mathbf{q}_{1}^{\dag}}
		\frac{\mathbf{q}_{1}^{\dag}\mathbf{p}_{1} \mathbf{q}_{2}^{\dag} \mathbf{p}_{1} \mathbf{q}_{1}^{\dag} \mathbf{p}_{2}}{
		\alpha_{1} \beta_{2} \mathbf{q}_{2}^{\dag} \mathbf{p}_{1}  \mathbf{q}_{1}^{\dag} \mathbf{p}_{2}- 
		 \alpha_{2} \beta_{1} \mathbf{q}_{1}^{\dag} \mathbf{p}_{1} \mathbf{q}_{2}^{\dag} \mathbf{p}_{2}},\\
		\beta_{4} &= \alpha_{2} (\beta_{1} + \beta_{2})(\alpha_{1} + \beta_{1})
		[\mathbf{Q}_{1}^{\dag}]_{\mathbf{p}_{1}}[\mathbf{P}_{2}]_{\mathbf{q}_{2}^{\dag}}
		\frac{- \mathbf{q}_{1}^{\dag}\mathbf{p}_{1} \mathbf{q}_{2}^{\dag} \mathbf{p}_{1} \mathbf{q}_{1}^{\dag} \mathbf{p}_{2}}{
		\alpha_{1} \beta_{2} \mathbf{q}_{2}^{\dag} \mathbf{p}_{1}  \mathbf{q}_{1}^{\dag} \mathbf{p}_{2}- 
		 \alpha_{2} \beta_{1} \mathbf{q}_{1}^{\dag} \mathbf{p}_{1} \mathbf{q}_{2}^{\dag} \mathbf{p}_{2}},
	\end{align}
	which, together with the expressions
	\begin{align}
		[\mathbf{Q}_{1}^{\dag}]_{\mathbf{p}_{1}}[\mathbf{P}_{1}]_{\mathbf{q}_{1}^{\dag}}
		&= \frac{-1}{(\alpha_{1} + \alpha_{2})(\alpha_{1} + \beta_{1})}\cdot
		\frac{\alpha_{1}(\beta_{2} - \kappa) \mathbf{q}_{2}^{\dag} \mathbf{p}_{1} \mathbf{q}_{1}^{\dag} \mathbf{p}_{2} 
		 - \alpha_{2} \beta_{1}\mathbf{q}_{1}^{\dag} \mathbf{p}_{1} \mathbf{q}_{2}^{\dag} \mathbf{p}_{2}}{
		\mathbf{q}_{1}^{\dag}\mathbf{p}_{1} \mathbf{q}_{2}^{\dag} \mathbf{p}_{1} \mathbf{q}_{1}^{\dag} \mathbf{p}_{2}}\\
		[\mathbf{Q}_{1}^{\dag}]_{\mathbf{p}_{1}}[\mathbf{P}_{2}]_{\mathbf{q}_{2}^{\dag}}
		&= \frac{1}{(\beta_{1} + \beta_{2})(\alpha_{1} + \beta_{1})}\cdot
		\frac{\alpha_{1} \beta_{2}  \mathbf{q}_{2}^{\dag} \mathbf{p}_{1} \mathbf{q}_{1}^{\dag} \mathbf{p}_{2} 
		 -(\alpha_{2} - \kappa) \beta_{1}\mathbf{q}_{1}^{\dag} \mathbf{p}_{1} \mathbf{q}_{2}^{\dag} \mathbf{p}_{2} }{
		\mathbf{q}_{1}^{\dag}\mathbf{p}_{1} \mathbf{q}_{2}^{\dag} \mathbf{p}_{1} \mathbf{q}_{2}^{\dag} \mathbf{p}_{2}},		
	\end{align}
	where $\kappa = \alpha_{1} + \alpha_{2} + \beta_{1} + \beta_{2}$, gives
	\begin{align}
		\alpha_{4} &= - \beta_{2} + 
		\kappa\frac{\alpha_{1}\beta_{2} \mathbf{q}_{2}^{\dag} \mathbf{p}_{1} \mathbf{q}_{1}^{\dag} \mathbf{p}_{2} }{
		\alpha_{1} \beta_{2} \mathbf{q}_{2}^{\dag} \mathbf{p}_{1}  \mathbf{q}_{1}^{\dag} \mathbf{p}_{2}- 
		 \alpha_{2} \beta_{1} \mathbf{q}_{1}^{\dag} \mathbf{p}_{1} \mathbf{q}_{2}^{\dag} \mathbf{p}_{2}
		},\\
		\beta_{4} &= - \alpha_{2} + 
		\kappa\frac{ - \alpha_{2}\beta_{1} \mathbf{q}_{1}^{\dag} \mathbf{p}_{1} \mathbf{q}_{2}^{\dag} \mathbf{p}_{2} }{
		\alpha_{1} \beta_{2} \mathbf{q}_{2}^{\dag} \mathbf{p}_{1}  \mathbf{q}_{1}^{\dag} \mathbf{p}_{2}- 
		 \alpha_{2} \beta_{1} \mathbf{q}_{1}^{\dag} \mathbf{p}_{1} \mathbf{q}_{2}^{\dag} \mathbf{p}_{2}},
	\end{align}
	and so it is easy to see that $\alpha_{4} + \beta_{4} = \kappa - \alpha_{2} - \beta_{2} = \alpha_{1} + \beta_{1}$,
	as required. Similarly, using 
	\begin{align}
		[\mathbf{Q}_{2}^{\dag}]_{\mathbf{p}_{2}}[\mathbf{P}_{1}]_{\mathbf{q}_{1}^{\dag}}
		&= \frac{1}{(\alpha_{1} + \alpha_{2})(\alpha_{2} + \beta_{2})}\cdot
		\frac{\alpha_{1}\beta_{2} \mathbf{q}_{2}^{\dag} \mathbf{p}_{1} \mathbf{q}_{1}^{\dag} \mathbf{p}_{2} 
		 - \alpha_{2} (\beta_{1} - \kappa)\mathbf{q}_{1}^{\dag} \mathbf{p}_{1} \mathbf{q}_{2}^{\dag} \mathbf{p}_{2}}{
		\mathbf{q}_{1}^{\dag}\mathbf{p}_{1} \mathbf{q}_{1}^{\dag} \mathbf{p}_{2} \mathbf{q}_{2}^{\dag} \mathbf{p}_{2}}\\
		[\mathbf{Q}_{2}^{\dag}]_{\mathbf{p}_{2}}[\mathbf{P}_{2}]_{\mathbf{q}_{2}^{\dag}}
		&= \frac{-1}{(\beta_{1} + \beta_{2})(\alpha_{2} + \beta_{2})}\cdot
		\frac{(\alpha_{1} - \kappa)\beta_{2}\mathbf{q}_{2}^{\dag} \mathbf{p}_{1} \mathbf{q}_{1}^{\dag} \mathbf{p}_{2} 
		 - \alpha_{2} \beta_{1}  \mathbf{q}_{1}^{\dag} \mathbf{p}_{1} \mathbf{q}_{2}^{\dag} \mathbf{p}_{2} }{
		\mathbf{q}_{2}^{\dag}\mathbf{p}_{1} \mathbf{q}_{1}^{\dag} \mathbf{p}_{2} \mathbf{q}_{2}^{\dag} \mathbf{p}_{2}},		
	\end{align}
	and the equations at the node $\mathbf{Q}_{2}^{\dag}$ gives 
	\begin{align}
		\alpha_{5} &= - \beta_{1} + 
		\kappa \frac{ - \alpha_{2}\beta_{1} \mathbf{q}_{1}^{\dag} \mathbf{p}_{1} \mathbf{q}_{2}^{\dag} \mathbf{p}_{2} }{
		\alpha_{1} \beta_{2} \mathbf{q}_{2}^{\dag} \mathbf{p}_{1}  \mathbf{q}_{1}^{\dag} \mathbf{p}_{2}- 
		 \alpha_{2} \beta_{1} \mathbf{q}_{1}^{\dag} \mathbf{p}_{1} \mathbf{q}_{2}^{\dag} \mathbf{p}_{2}},\\
		\beta_{5} &= - \alpha_{1} + 
		\kappa\frac{\alpha_{1}\beta_{2} \mathbf{q}_{2}^{\dag} \mathbf{p}_{1} \mathbf{q}_{1}^{\dag} \mathbf{p}_{2} }{
		\alpha_{1} \beta_{2} \mathbf{q}_{2}^{\dag} \mathbf{p}_{1}  \mathbf{q}_{1}^{\dag} \mathbf{p}_{2}- 
		 \alpha_{2} \beta_{1} \mathbf{q}_{1}^{\dag} \mathbf{p}_{1} \mathbf{q}_{2}^{\dag} \mathbf{p}_{2}
		}.
	\end{align}
	Similar computations show that 
	equations at the nodes $\mathbf{P}_{1}$ and $\mathbf{P}_{2}$ give the same values of the parameters, and the 
	other constraints are easy to check.
\end{proof}

\begin{lemma}[The \emph{Edge} Lemma]\label{lem:edge-lemma} Let $(\mathbf{p}_{1}, \mathbf{p}_{2}, \mathbf{p}_{3})$
	be a $\mathbf{q}^{\dag}_{3}$-triple with parameters $(\alpha_{1},\alpha_{2},\alpha_{3})$, 
	$(\mathbf{q}^{\dag}_{1},\mathbf{q}^{\dag}_{2},\mathbf{q}^{\dag}_{3})$ be a $\mathbf{p}_{3}$-triple
	with parameters $(\beta_{1},\beta_{2},\beta_{3})$ with parameters $(\beta_{1}, \beta_{2}, \beta_{3})$, and let 
	$\alpha_{3} = \beta_{3}$,
	\begin{center}
		\includegraphics[height=1.2in]{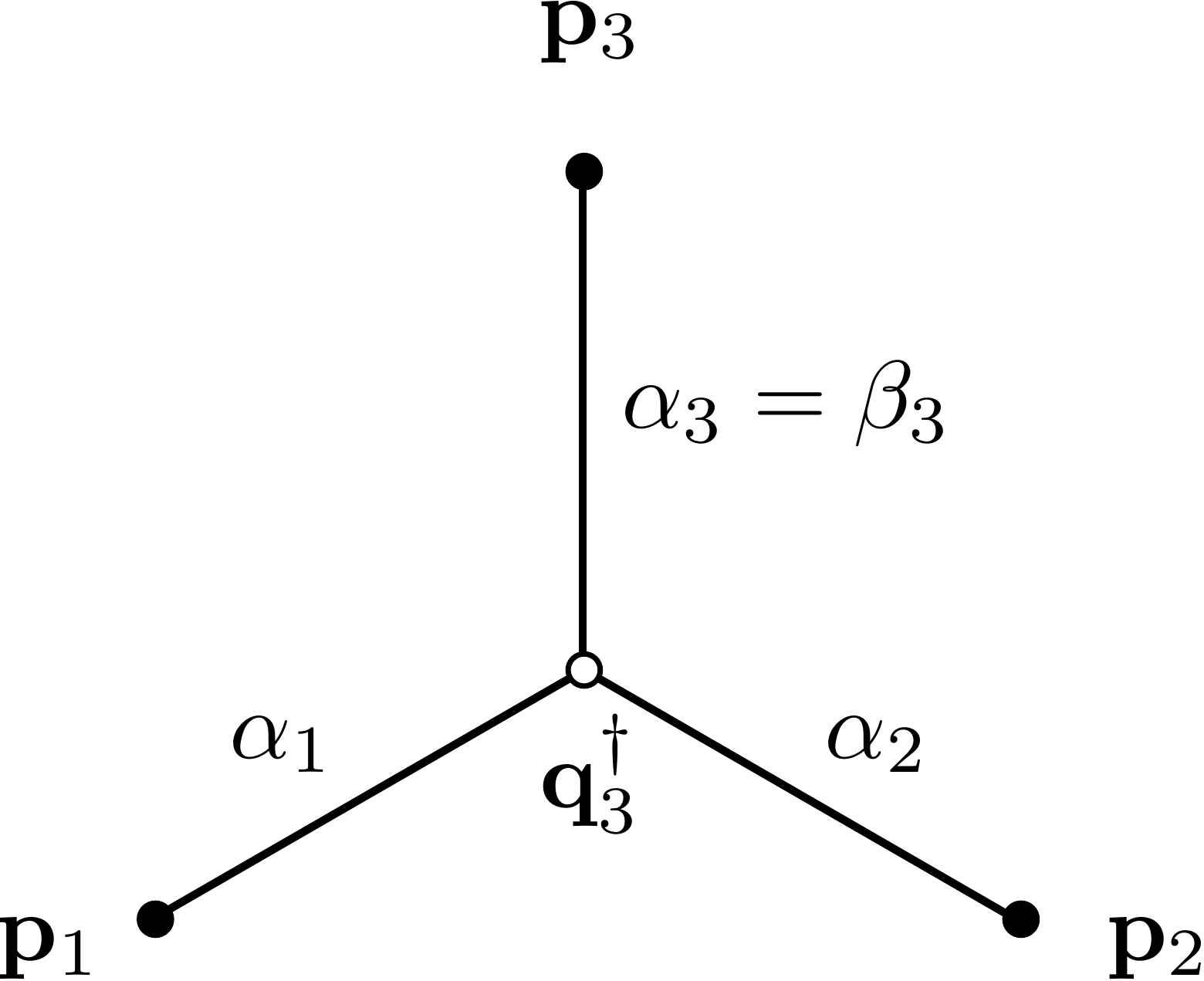}\qquad \includegraphics[height=1.2in]{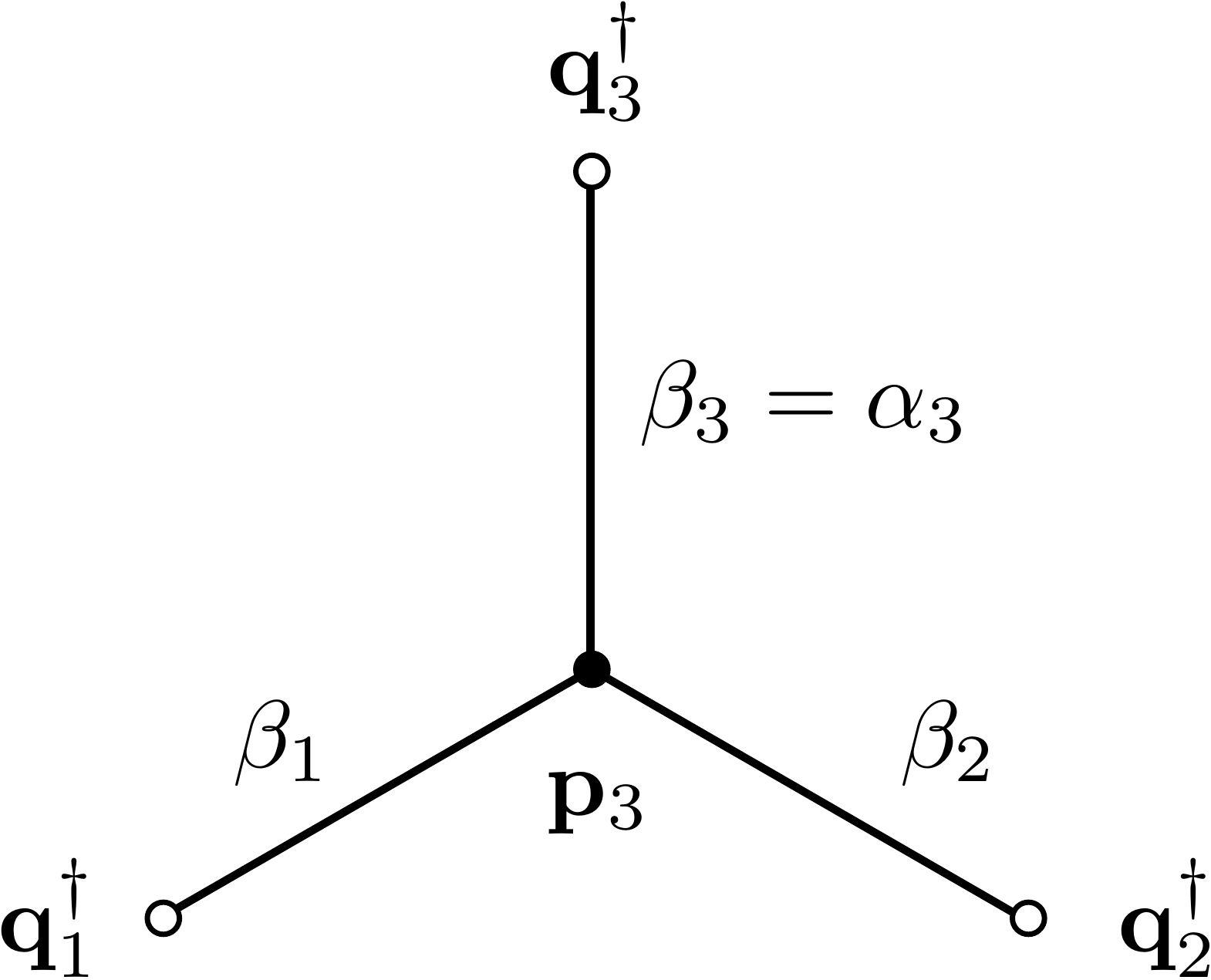}.
	\end{center}
	Then it is possible to glue the triples along the $\mathbf{p}_{3} \mathbf{q}^{\dag}_{3}$-edge and, by assigning
	parameters $\gamma_{1}$ and $\gamma_{2}$ along the $\mathbf{p}_{1}\mathbf{q}^{\dag}_{1}$ and 
	$\mathbf{p}_{2}\mathbf{q}^{\dag}_{2}$ edges so that the sum of parameters in the new quads is zero (i.e.,
	$\gamma_{1} = \alpha_{2} - \beta_{1}$ and $\gamma_{2} = \alpha_{1} - \beta_{2}$), complete the resulting
	``butterfly'' configuration to a consistent cube,
	\begin{center}
		\includegraphics[height=1.5in]{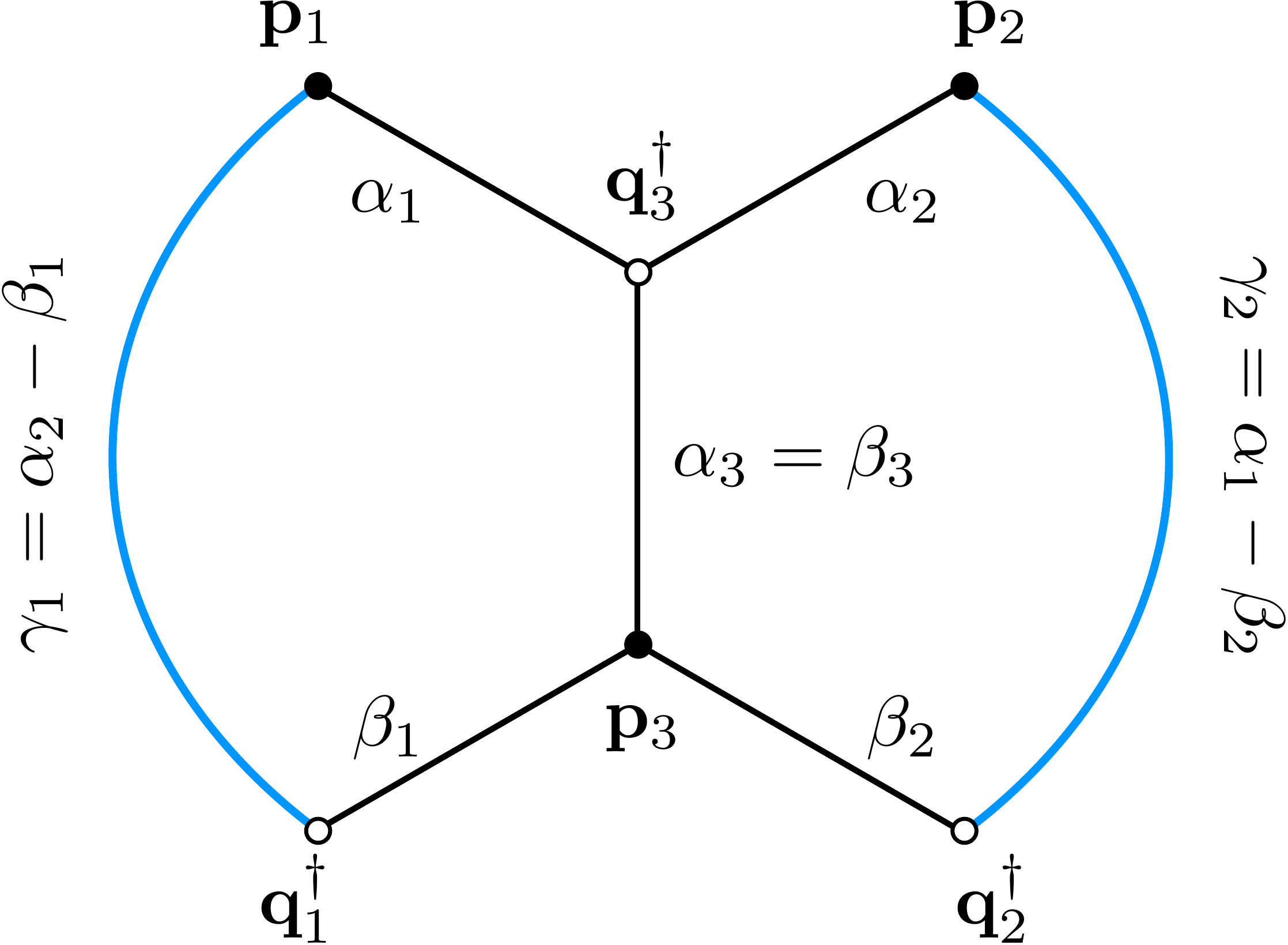}\qquad \includegraphics[height=1.5in]{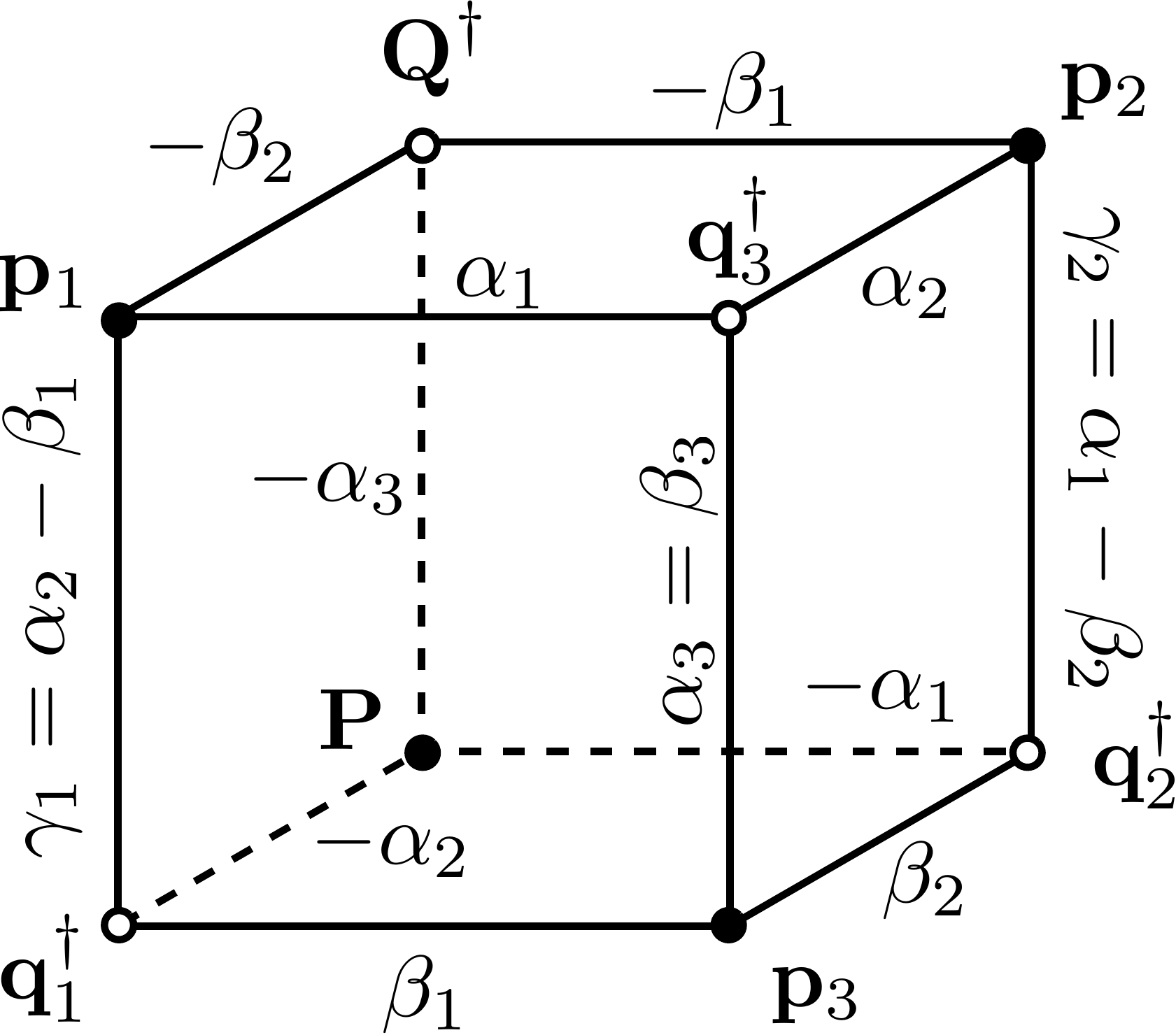}.
	\end{center}
\end{lemma}
\begin{proof} After choosing $\gamma_{1}$ and $\gamma_{2}$ according to the above rule, 
	this follows immediately from the \emph{Face Lemma} \ref{lem:face-lemma}.
	
\end{proof}
% subsection gluing_properties (end)

\section{Eigenvectors of Quadratic Lax Matrices and Refactorization} % (fold)
\label{sec:eigenvectors_of_quadratic_lax_matrices_and_refactorization}

In this section we show that a cube diagram described in the previous section can be used to 
represent the relationship between eigenvectors of quadratic Lax matrices. Also, we  
show that the three axes of symmetry of the cube connecting the centers of opposing sides 
can be give the following interpretation. One of the axis corresponds to switching the 
coordinate system on the space from \emph{left} to \emph{right} divisors. The other two axes
correspond to two different directions of the refactorization dynamics corresponding to two possible pairings
between two zeroes and two poles of $\det \mathbf{L}(z)$. For each of this transformations
we produce the explicit formula for the generating function using the labeling of the cube.

\subsection{Relations Between Eigenvectors of Quadratic Lax Matrices} % (fold)
\label{sub:relations_between_eigenvectors_of_quadratic_lax_matrices}

Consider a quadratic Lax matrix $\mathbf{L}(z)$ given in the additive form as 
\begin{align}
	\mathbf{L}(z) &= \mathbf{L}_{0} + \frac{ \mathbf{a}_{1} \mathbf{b}_{1}^{\dag} }{ z - z_{1} } + 
	\frac{ \mathbf{a}_{2} \mathbf{b}_{2}^{\dag} }{ z - z_{2} },\quad
	\mathbf{L}(z)^{-1} = \mathbf{L}_{0}^{-1} - \frac{ \mathbf{c}_{\alpha} \mathbf{d}_{\alpha}^{\dag} }{ z - \zeta_{\alpha} }
	- \frac{ \mathbf{c}_{\beta} \mathbf{d}_{\beta}^{\dag} }{ z - \zeta_{\beta} },\\
	\det \mathbf{L}(z) &= \det \mathbf{L}_{0} \frac{ (z - \zeta_{\alpha})(z - \zeta_{\beta}) }{ (z - z_{1})(z- z_{2}) }.
\end{align}

Pairing $z_{1}$ with $\zeta_{\alpha} = \zeta_{1}$ and $z_{2}$ with $\zeta_{\beta} = \zeta_{2}$ we can write $\mathbf{L}(z)$ and 
$\mathbf{L}^{-1}(z)$ in the multiplicative form using factors 
$\mathbf{B}_{i}(z) = \mathbf{I} + \frac{ z - \zeta_{i} }{ z - z_{i} } \frac{ \mathbf{p}_{i} \mathbf{q}_{i}^{\dag} }{ \mathbf{q}_{i}^{\dag} \mathbf{p}_{i} }$,
\begin{equation}
	\mathbf{L}(z) = \mathbf{L}_{0} \mathbf{B}_{1}(z) \mathbf{B}_{2}(z),\qquad \mathbf{L}(z)^{-1} = \mathbf{B}_{2}(z)^{-1} \mathbf{B}_{1}(z)^{-1}\mathbf{L}_{0}^{-1},
\end{equation}
and so $\mathbf{B}_{1}(z) = \mathbf{L}_{0}^{-1} \mathbf{B}_{1}^{l}(z) \mathbf{L}_{0} = \mathbf{I} + \frac{ z_{1} - \zeta_{\alpha} }{ z - z_{1} } 
\frac{ \mathbf{L}_{0}^{-1} \mathbf{a}_{1} \mathbf{d}_{\alpha}^{\dag} \mathbf{L}_{0}}{ \mathbf{d}_{\alpha}^{\dag} \mathbf{a}_{1} }$ 
and $\mathbf{B}_{2}(z) = \mathbf{B}_{2}^{r}(z) = \mathbf{I} + \frac{ z_{2} - \zeta_{\beta} }{ z - z_{2} } 
\frac{ \mathbf{c}_{\beta} \mathbf{b}_{2}^{\dag}}{ \mathbf{b}_{2}^{\dag} \mathbf{c}_{\beta} }$. This can also be seen directly by looking at 
the residues of $\mathbf{L}(z)$ at $z_{1}$ and $z_{2}$ and of $\mathbf{L}(z)^{-1}$ at $\zeta_{\alpha}$ and $\zeta_{\beta}$; in fact, that is 
how the expressions for the right and left divisors are obtained. The same residues
also give the following collection of equations:
\begin{equation}
	\mathbf{d}_{\alpha}^{\dag} \mathbf{L}_{0} \mathbf{B}_{2}(z_{1}) \sim \mathbf{b}_{1}^{\dag},\quad
	\mathbf{B}_{1}(z_{2}) \mathbf{c}_{\beta} \sim \mathbf{L}_{0}^{-1} \mathbf{a}_{2},\quad
	\mathbf{B}_{2}(\zeta_\alpha) \mathbf{c}_{\alpha} \sim \mathbf{L}_{0}^{-1} \mathbf{a}_{1},\quad
	\mathbf{d}_{\beta}^{\dag} \mathbf{L}_{0} \mathbf{B}_{1}(\zeta_{\beta}) \sim \mathbf{b}_{2}^{\dag}.
\end{equation}
Representing these equations using elementary triples and gluing, we get the following result.

\begin{theorem}
	\label{thm:refactor-2-mult}
	The eigenvectors of $\mathbf{L}(z)$ are related by the following cube diagram:
	\begin{center}
		\includegraphics[width=3in]{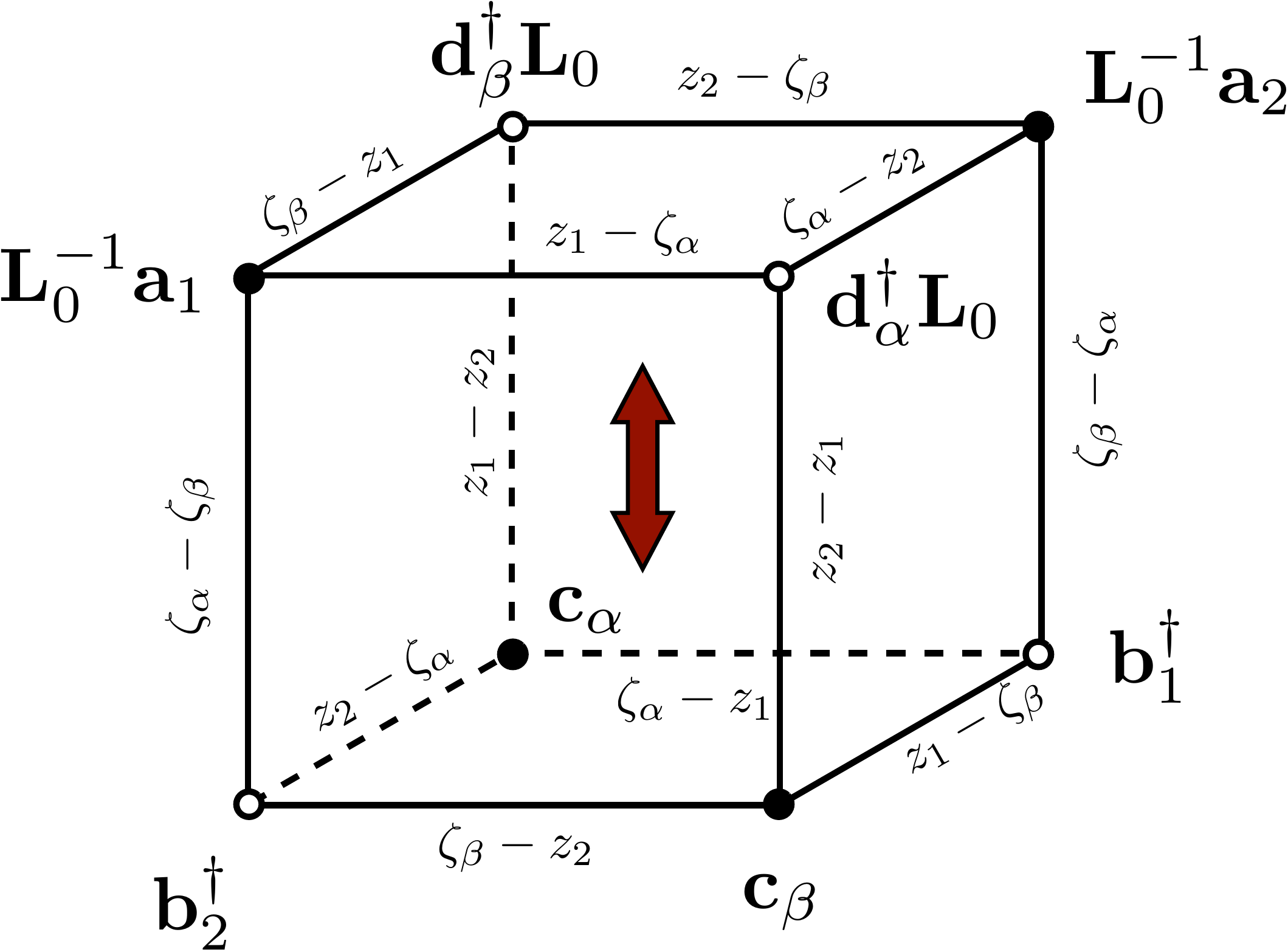}.
	\end{center}
\end{theorem}

Note that in this diagram, the \emph{top} (resp.~\emph{bottom}) faces of the cube correspond to parameterizing the space 
of quadratic Lax matrices by \emph{left} (resp.~\emph{right}) divisors, and so the vertical axis of the cube correspond to 
the change of coordinates between these two different systems. Two horizontal axes correspond to the refactorization dynamics.
% subsection relations_between_eigenvectors_of_quadratic_lax_matrices (end)

\subsection{Refactorization Dynamics} % (fold)
\label{sub:refactorization_dynamics}
Consider now an isospectral dynamic $\mathbf{L}(z)\mapsto \tilde{\mathbf{L}}(z) = \mathbf{R}(z) \mathbf{L}(z) \mathbf{R}(z)^{-1}$,
the isomonodromic case is similar. If we take $\mathbf{R}(z) = \mathbf{B}_{2}(z) =\mathbf{B}^{r}_{\beta;2}(z)$, where the notation
$\mathbf{B}^{r}_{\beta;2}(z)$ explicitly specifies the zero and the pole of the elementary divisor, this becomes a refactorization transformation
\begin{equation}
	\mathbf{L}(z) = \mathbf{L}_{0} \mathbf{B}_{\alpha;1}(z) \mathbf{B}_{\beta;2}(z) \mapsto 
	\tilde{\mathbf{L}}(z) = \mathbf{B}_{\beta;2}(z) \mathbf{L}_{0} \mathbf{B}_{\alpha;1}(z) = 
	\mathbf{L}_{0} \tilde{\mathbf{B}}_{\alpha;1}(z) \tilde{\mathbf{B}}_{\beta;2}(z).
\end{equation}

Since we can think of the refactorization transformation as switching the roles of the left and right divisors,
$\mathbf{B}_{\beta,2}(z)  = \mathbf{B}_{\beta,2}^{r}(z) = \tilde{\mathbf{B}}_{\beta,2}^{l}(z)$ and 
$\mathbf{B}_{\alpha,1}(z) = \mathbf{L}_{0}^{-1} \mathbf{B}_{\alpha,1}^{l}(z) \mathbf{L}_{0} = \tilde{\mathbf{B}}_{\alpha,1}^{r}(z)$,
using (\ref{eq:left-right-divs}) we get the identifications $\tilde{\mathbf{a}}_{2} = \mathbf{c}_{\beta}$,
$\tilde{\mathbf{d}}_{\beta}^{\dag} = \mathbf{b}_{2}$, $\tilde{\mathbf{c}}_{\alpha} = \mathbf{L}_{0}^{-1}\mathbf{a}_{1}$, and 
$\tilde{\mathbf{b}}^{\dag}_{1} = \mathbf{d}_{\alpha}^{\dag} \mathbf{L}_{0}$. Thus, on the diagram below labels on the back face
correspond to the coordinates of $\mathbf{L}(z)$, labels on the front face correspond (after twisting $\tilde{\mathbf{a}}_{2}$ by
$\mathbf{L}_{0}^{-1}$ and $\tilde{\mathbf{d}}_{\beta}^{\dag}$ by $\mathbf{L}_{0}$) to the coordinates of $\tilde{\mathbf{L}}(z)$,
	\begin{center}
		\includegraphics[width=3in]{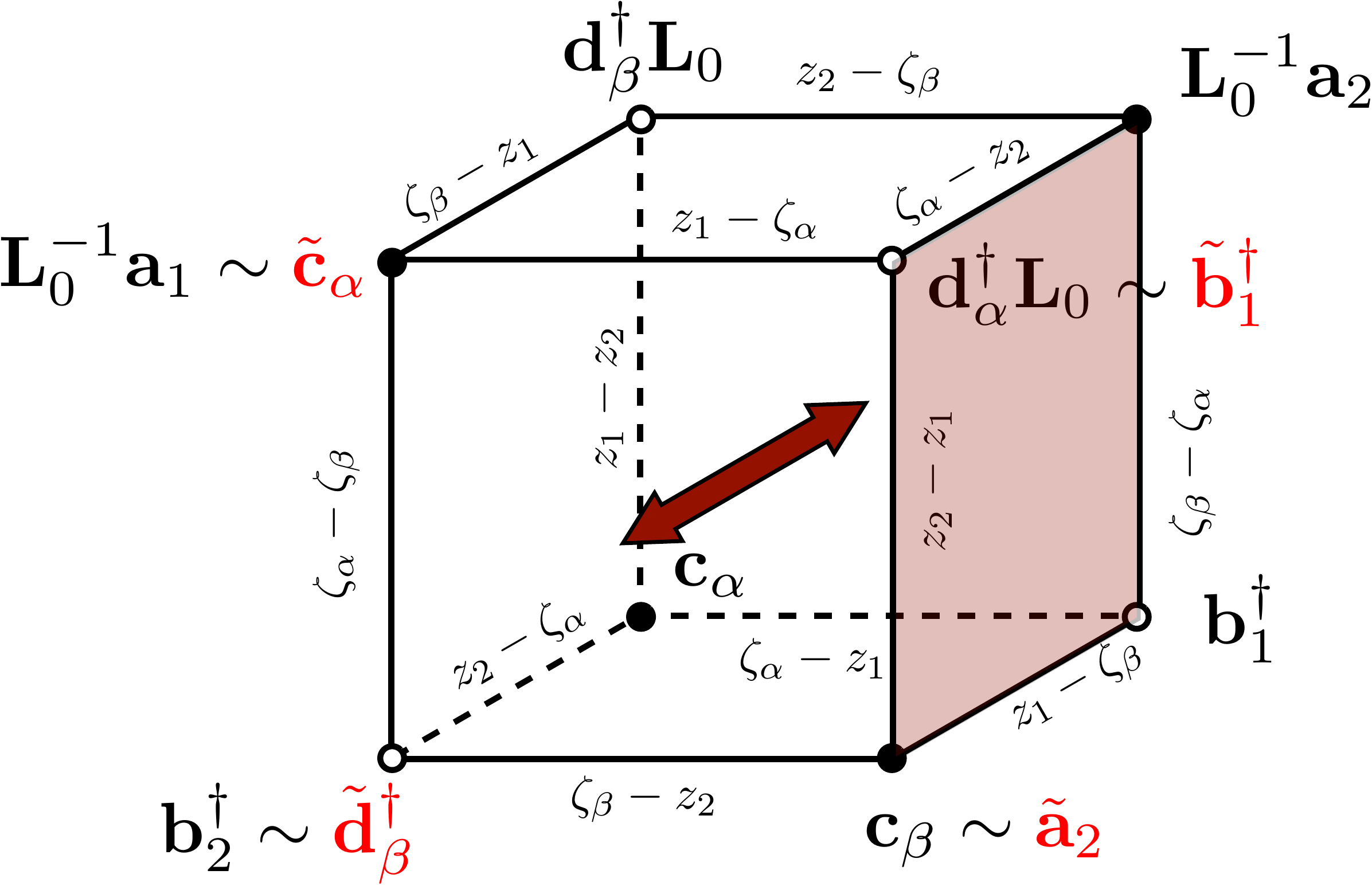},
	\end{center}
and the generating function encoded by the \emph{right} face of the cube is, up to some minor change of notation, 
exactly the Lagrangian from Theorem~(\ref{thm:Lagrangian}),
\begin{align}
	\mathcal{L}((\mathbf{a}_{2},\mathbf{b}^{\dag}_{1}),
	\widetilde{(\mathbf{a}_{2},\mathbf{b}^{\dag}_{1})}) &= 
	(z_{2} - z_{1})\log(\tilde{\mathbf{b}}^{\dag}_{1} \tilde{\mathbf{a}}_{2}) + 
	(z_{1} - \zeta_{\beta})\log(\mathbf{b}^{\dag}_{1} \tilde{\mathbf{a}}_{2}) \\
	&\qquad +
	(\zeta_{\beta} - \zeta_{\alpha}) \log(\mathbf{b}^{\dag}_{1} \mathbf{L}_{0}^{-1}\mathbf{a}_{2}) + 
	(\zeta_{\alpha} - z_{2}) \log(\tilde{\mathbf{b}}^{\dag}_{1} \mathbf{L}_{0}^{-1}\mathbf{a}_{2}).\notag
\end{align}
Note that using the generating function given by the \emph{left} face of the cube corresponds to the backwards motion generated by 
$\mathbf{R}(z) = \mathbf{B}^{r}_{\alpha;1}(z)$.

The remaining axis of the cube corresponds to pairing $z_{1}$ with $\zeta_{\beta}$ and $z_{2}$ with $\zeta_{\alpha}$. In other words, 
taking $\mathbf{R}(z) = \mathbf{B}^{r}_{\alpha;2}(z)$ (or $\mathbf{R}(z) = \mathbf{B}^{r}_{\beta;1}(z)$ for the backwards motion),
\begin{equation}
	\mathbf{L}(z) = \mathbf{L}_{0} \mathbf{B}_{\beta;1}(z) \mathbf{B}_{\alpha;2}(z) \mapsto 
	\tilde{\mathbf{L}}(z) = \mathbf{B}_{\alpha;2}(z) \mathbf{L}_{0} \mathbf{B}_{\beta;1}(z) = 
	\mathbf{L}_{0} \tilde{\mathbf{B}}_{\beta;1}(z) \tilde{\mathbf{B}}_{\alpha;2}(z),
\end{equation}
we get the following map
	\begin{center}
		\includegraphics[width=3in]{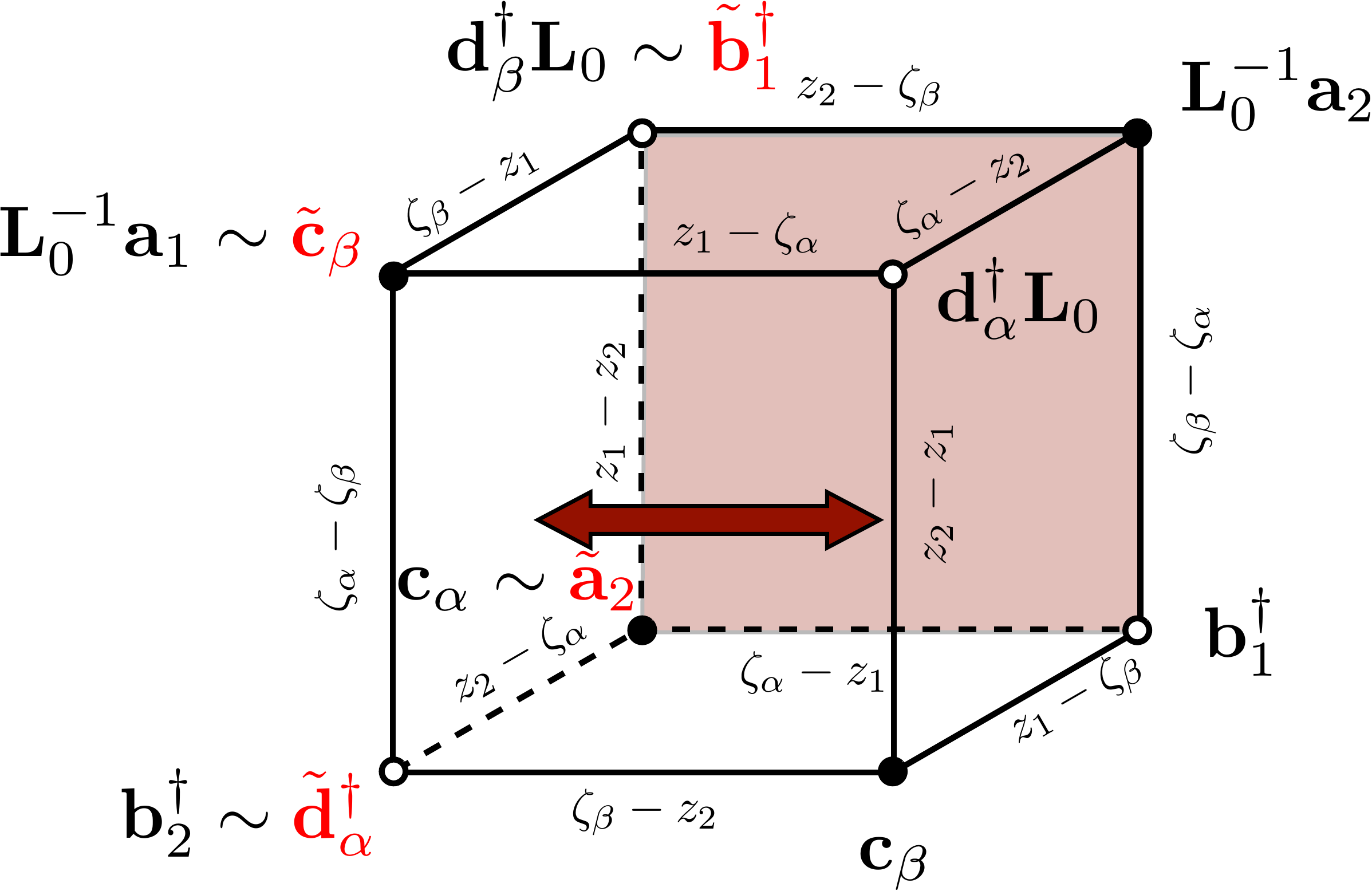}.
	\end{center}
and its Lagrangian generating function
\begin{align}
	\mathcal{L}((\mathbf{a}_{2},\mathbf{b}^{\dag}_{1}),
	\widetilde{(\mathbf{a}_{2},\mathbf{b}^{\dag}_{1})}) &= 
	(z_{1} - z_{2})\log(\tilde{\mathbf{b}}^{\dag}_{1} \tilde{\mathbf{a}}_{2}) + 
	(\zeta_{\alpha} - z_{1})\log(\mathbf{b}^{\dag}_{1} \tilde{\mathbf{a}}_{2}) \\
	&\qquad +
	(\zeta_{\beta} - \zeta_{\alpha}) \log(\mathbf{b}^{\dag}_{1} \mathbf{L}_{0}^{-1}\mathbf{a}_{2}) + 
	(z_{2} - \zeta_{\beta}) \log(\tilde{\mathbf{b}}^{\dag}_{1} \mathbf{L}_{0}^{-1}\mathbf{a}_{2}).\notag
\end{align}

% subsection refactorization_dynamics (end)

% section eigenvectors_of_quadratic_lax_matrices_and_refactorization (end)

\section{Conclusions} % (fold)
\label{sec:conclusions}

We explained a neat and efficient way to encode the structure of refactorization transformations for quadratic Lax matrices and their generating
functions using cube diagrams. This approach also gives different
ways of choosing coordinate systems on the space of such matrices --- each choice corresponds to a face of the cube. 
It would be very interesting to see 
if this approach can be generalized to Lax matrices with more than two factors. Since, in view of  the gluing 
properties, cube diagrams are \emph{rigid}, we expect it to results in higher dimensional configurations relating 
such cubes, where each cube represents a particular transposition of factors. In particular, the case of three factors 
is related to the structure of Yang-Baxter maps and we plan to consider it in a separate publication.
% section conclusions (end)

\section{Acknowledgements} % (fold)
\label{sec:acknowledgements} The author thanks Adam Doliwa, Michael Gekhtman, and Yuri Suris for helpful conversations and suggestions.

% section acknowledgements (end)

\small
\bibliographystyle{amsalpha}

% \bibliography{../../Bibliography/adzham}

\providecommand{\bysame}{\leavevmode\hbox to3em{\hrulefill}\thinspace}
\providecommand{\MR}{\relax\ifhmode\unskip\space\fi MR }
\providecommand{\MRhref}[2]{%
  \href{http://www.ams.org/mathscinet-getitem?mr=#1}{#2}
}
\providecommand{\href}[2]{#2}

\end{document}